\documentclass[review,onefignum,onetabnum]{siamart220329}

\usepackage{lipsum}
\usepackage{amsfonts}
\usepackage{graphicx}
\usepackage{algorithmic}
\ifpdf
  \DeclareGraphicsExtensions{.eps,.pdf,.png,.jpg}
\else
  \DeclareGraphicsExtensions{.eps}
\fi

\usepackage{cancel}

\usepackage{hhline}
\usepackage{booktabs}
\usepackage{multirow}
\usepackage{siunitx}
\usepackage{float}
\usepackage{caption}
\usepackage{subcaption}
\usepackage[export]{adjustbox}
\usepackage{enumitem}

\newsiamremark{remark}{Remark}
\newsiamremark{hypothesis}{Hypothesis}
\crefname{hypothesis}{Hypothesis}{Hypotheses}
\newsiamthm{claim}{Claim}

\headers{Error Estimate of Energy Stable Flux Reconstruction Method}{E. Lambert, and S. Nadarajah}

\title{An $\mathcal{L}^2$-Error Estimate of Energy Stable Flux Reconstruction Method
\thanks{\funding{This work was funded by Natural Sciences and Engineering Research Council of Canada Discovery Grant Program RGPIN-2019-04791}}}

\author{Erwan Lambert\thanks{Computational aerodynamic group, Department of Mechanical Engineering, McGill University, Montreal
  (\email{erwan.lambert@mail.mcgill.ca},\email{siva.nadarajah@mcgill.ca}).}
\and Siva Nadarajah\footnotemark[2]}

\usepackage{amsopn}

\makeatletter
\newcommand*{\addFileDependency}[1]{
  \typeout{(#1)}
  \@addtofilelist{#1}
  \IfFileExists{#1}{}{\typeout{No file #1.}}
}
\makeatother

\begin{document}

\maketitle

\begin{abstract}
Energy stable flux reconstruction (ESFR) is a high-order numerical method used for solving partial differential equations in computational fluid dynamics. This method is designed to preserve the energy stability of the underlying partial differential equation system with respect to a broken Sobolev norm. A class of one-parameter ESFR schemes has been identified to be stable for the one-dimensional linear advection equation. This class includes some well-known high-order methods such as the discontinuous Galerkin method and spectral difference method. The main advantage of the energy stable flux reconstruction is to allow for an increase in the maximum admissible time step while retaining the stability and accuracy properties of the underlying scheme.
However numerical experiments have shown that beyond a certain value of the parameter, the optimal order of accuracy is lost. 
This article develops an $L^2$-error estimate for the energy stable flux reconstruction scheme applied to the one-dimensional advection equation and demonstrates the exact expression that contributes towards the loss of the optimal order.
\end{abstract}



\section{Introduction}

Computational fluid dynamics (CFD) methods, which are crucial for simulating fluid flow in engineering systems, can be categorized into low-order and high-order approaches.

While low-order methods are robust and versatile, they suffer from limitations such as high numerical dissipation~\cite{spiegel2015survey} and poor parallel scalability.

High-order methods offer a significant advantage with their low numerical dissipation, making them suitable for time-dependent problems like turbulent flows. They excel in capturing vortex structures and resolving fine scales in fluid dynamics applications~\cite{vincent2011facilitating}. 
Despite their growing popularity, high-order methods can be computationally expensive due to time-stepping restrictions. Efforts are ongoing to reduce costs by increasing the allowable time step while maintaining accuracy.

High-order methods are often designed to provide superior precision to low-order schemes with similar computational costs.
One commonly used family of high-order methods is the discontinuous Galerkin (DG) method, originally proposed by Reed and Hill in 1973~\cite{reed1973triangular}. DG methods decompose the numerical solution into piecewise discontinuous polynomials within each element, allowing for spatial and spectral resolution refinement. Later, Kopriva et al. proposed the staggered grid method~\cite{kopriva1996conservative}, later renamed as spectral difference (SD) by Wang et al.~\cite{wang2007spectral}, which is formulated based on the differential form of the equation. More recently Huynh introduced the flux reconstruction (FR) approach that unifies various high-order schemes, including collocation-based nodal discontinuous Galerkin (DG) methods, and some spectral difference methods~\cite{huynh2007flux}. Vincent, Castonguay, and Jameson generalized this framework and identified an infinite range of linearly stable FR schemes~\cite{vincent2011new}, later named energy stable flux reconstruction (ESFR).
Other important properties of ESFR schemes such as dispersion and dissipation properties have been investigated using Von Neumann analysis~\cite{vincent2011insights}. They also studied the explicit time step limits associated with ESFR schemes and were able to find some schemes for which the maximum admissible time step is significantly higher than the one for well-known schemes DG~\cite{hesthaven2007nodal}, SD~\cite{kopriva1996conservative} and Huynh~\cite{huynh2007flux}. However, DG yields the smallest error among the methods compared. 
Castonguay presented some numerical studies on the order of convergence for the one-dimensional linear advection~\cite{Castonguay}, where for a $k$-th order ESFR scheme he obtained either $k+1$- or $k$-th order of convergence depending on the value of the ESFR parameter, commonly referred to as the $c$ parameter. However, contrary to the DG method, there is no analytical proof to support this result.

{\color{black}An \textit{a priori} error estimate is a crucial tool in numerical analysis to assess the accuracy of numerical methods. The error estimate of DG methods in terms of the mesh size parameter has been investigated especially by Cockburn ~\cite{cockburn1999discontinuous}. The theorem 3 of Cockburn's book shows that the DG method is ($k+1$)-th order accurate, at least in the linear case. However, there is no existing work on proving the error estimate for the ESFR method. Our objective is to prove a similar order of accuracy for ESFR applied to the linear advection problem, provided that the exact solution is smooth enough.Given the close connection of the flux reconstruction method with Sobolev space, it is intuitive to pursue an error estimate within a norm defined in this domain. The most straightforward choice in this regard is the $\mathcal{L}^2$-norm.}

{\color{black}
The paper is organized as follows. Section~\ref{chap:FR} provides a comprehensive review of the flux reconstruction framework as it was introduced in~\cite{vincent2011new} and the derivation of the bilinear form for the ESFR applied to the linear advection equation. Section~\ref{chap:error_estimate} focuses on the development of an error estimate for the ESFR in the linear case, building on the outline of Cockburn's proof. In Section~\ref{chap:results}, the effectiveness of the proposed error estimate is discussed. Furthermore, the numerical experiments are extended to the one-dimensional inviscid Burgers' equations and two-dimensional Euler equations, and conclusions follow in Section~\ref{chap:conclusion}.}

\section{Flux Reconstruction Framework}
\label{chap:FR}

Flux reconstruction (FR) schemes share many similarities with nodal discontinuous Galerkin (DG) methods. Both FR and nodal DG schemes employ a high-order polynomial basis to approximate the solution within each element of the domain, and neither requires solution continuity between the elements.

\subsection{Preliminaries}

Let us consider a one-dimensional scalar conservation law:
\begin{align}
    &\partial_tu+\partial_xf(u)=0\text{ in }\Omega\times[0,T]\label{eq:pb},\\
    &u(x,0)=u_0(x)\hspace{3mm}\forall x\in \Omega,
\end{align}
with periodic boundary conditions, $u$ is the scalar conserved quantity and $f$ is the flux of $u$ in the $x$-direction.
Then $\Omega$ is discretized into $N$ non overlapping intervals $I_j=[x_{j-1/2},x_{j+1/2}]$, such that:
   $\Omega=\bigcup\limits_{j=1}^NI_j\hspace{1cm}\bigcap\limits_{j=1}^NI_j=\O.$

The FR scheme requires that $u$ is approximated on each sub-interval $I_j$, by a function $u^\delta_j=u^\delta_j(x,t)$ which is a polynomial of degree $k$ within $I_j$. In addition, the flux $f$ is also required to be approximated on each $I_j$ by a function $f^\delta_j=f^\delta_j(x,t)$ which is a polynomial of degree $k+1$ within $I_j$. Then a total approximation of the solution and of the flux can be defined on $\Omega$:
\begin{equation}
    u^{\delta}=\sum_{j=1}^{N} u_{j}^{\delta} \approx u, \quad f^{\delta}=\sum_{j=1}^{N} f_{j}^{\delta} \approx f
\end{equation}

As in the DG method, there is no inter-element continuity requirement on $u^\delta$, but $f^\delta$ is required to be continuous at the element interfaces; hence common numerical fluxes are introduced to ensure conservation.

\subsection{Implementation}

Now, consider the mapping from $I_j$ to the standard element $\Omega_s=[-1;1]$ via:
\begin{equation}
    r=\Gamma_j(x)=2\left(\frac{x-x_{j-1/2}}{x_{j+1/2}-x_{j-1/2}}\right)-1,
    \label{eq:mapping}
\end{equation}
which has the inverse mapping:
\begin{equation}
    x=\Gamma_j^{-1}(r)=\left(\frac{1-r}{2}\right) x_{j-1/2}+\left(\frac{1+r}{2}\right) x_{j+1/2}.
    \label{eq:inverse_mapping}
\end{equation}

The evolution of $u^\delta$ can be determined by solving the transformed equation in the standard element:
\begin{equation}
    \partial_t\hat{u}^\delta+\partial_r\hat{f}^\delta=0,\text{ where}
    \label{eq:transformed_equation}
\end{equation}
\begin{minipage}{0.51\textwidth}
\begin{equation}
    \hat{u}^\delta=\hat{u}^\delta(r, t)=u_j^\delta\left(\Gamma_j^{-1}(r), t\right), \text{ and}
    \label{eq:u_transformed}
\end{equation}
\end{minipage}
\hfill 
\begin{minipage}{0.48\textwidth}
\begin{equation}
    \hat{f}^\delta=\hat{f}^\delta(r, t)=\frac{2}{\Delta_j}f_j^\delta\left(\Gamma_j^{-1}(r), t\right).
\end{equation}
\end{minipage}

The flux reconstruction method is usually decomposed into five stages. The first stage consists in representing $\hat{u}^\delta$ as a $k$-th degree polynomial via a nodal Lagrange polynomial basis $l_i=l_i(r)$:

\noindent
\begin{minipage}{0.4\textwidth}
\begin{equation}
    \hat{u}^\delta=\sum_{i=0}^k \hat{u}_i^\delta l_i,\quad\quad \text{ where }
    \label{eq:u_hat_delta}
\end{equation}
\end{minipage}
\hfill
\begin{minipage}{0.5\textwidth}
\begin{equation}
    l_i=\prod_{m=0, m \neq i}^k\left(\frac{r-r_m}{r_i-r_m}\right).
    \label{eq:Lagrange}
\end{equation}
\end{minipage}

The second stage involves the construction of the approximate discontinuous flux $\hat{f}^{\delta D}$ as a $k$ degree polynomial, through a projection at the $k+1$ quadrature points, to ensure that it can be expressed in the same nodal basis:
\begin{equation}
    \hat{f}^{\delta D}=\sum_{i=0}^k \hat{f}_i^{\delta D} l_i.
    \label{eq:f_hat_delta_D}
\end{equation}
$\hat{f}^{\delta D}$ is discontinuous in general since it is directly computed from $\hat{u}^{\delta}$ which is a piecewise polynomial and usually discontinuous at interfaces between elements.

The third stage consists in evaluating the approximate solution $\hat{u}^{\delta}$ and discontinuous flux $\hat{f}_i^{\delta D}$ at both edges of the element. The fourth stage involves building the total approximation of the flux $\hat{f}^{\delta}$ as a $k+1$ degree polynomial by adding a correction flux to the discontinuous flux:
\begin{equation}
    \hat{f}^{\delta C}=\left(\hat{f}_{L}^{\delta I}-\hat{f}_{L}^{\delta D}\right) g_{L}+\left(\hat{f}_{R}^{\delta I}-\hat{f}_{R}^{\delta D}\right) g_{R},
\end{equation}
where:
\setlist[itemize]{leftmargin=2em}
\begin{itemize}
\item $\hat{f}^{\delta I}=\hat{f}^{\delta I}(r,t)$ : interface flux = numerical flux
    \item $\hat{f}^{\delta D}=\hat{f}^{\delta D}(r,t)$ : Discontinuous flux computed from approximation solution $u^\delta$
    \item $g_L=g_L(r)$: Left correction function
    \item $g_R=g_R(r)$: Right correction function
    \item $\hat{f}^{\delta I}_L=\hat{f}^{\delta I}(-1)$
    \item $\hat{f}^{\delta I}_R=\hat{f}^{\delta I}(1)$
    \item \textit{idem} for $\hat{f}^{\delta D}_L$ and $\hat{f}^{\delta D}_R$
\end{itemize}

To ensure continuity, the correction functions satisfy the following boundary condition:
\begin{align}
&g_{L}(-1)=1, & g_{L}(1)=0, \label{eq:condgl} \\
&g_{R}(-1)=0, & g_{R}(1)=1.\label{eq:condgr}
\end{align}\vspace{-\baselineskip}

The total approximation of the flux is a $k+1$-order polynomial of the form: 
\begin{equation}
    \hat{f}^{\delta}=\hat{f}^{\delta D}+\hat{f}^{\delta C}=\hat{f}^{\delta D}+\left(\hat{f}_{L}^{\delta I}-\hat{f}_{L}^{\delta D}\right) g_{L}+\left(\hat{f}_{R}^{\delta I}-\hat{f}_{R}^{\delta D}\right) g_{R}.
\end{equation}

Huynh identified several correction functions that lead to stable schemes~\cite{huynh2007flux}, in particular correction functions corresponding to the discontinuous Galerkin (DG), the spectral difference (SD) as well as Huynh type schemes (HU), named after him by Vincent et al~\cite{vincent2011new}:
\begin{align}
&g_{DG_L}=\frac{(-1)^k}{2}\left(L_k-L_{k+1}\right)\text{ and }g_{DG_R}=\frac{1}{2}\left(L_k+L_{k+1}\right) \label{eq:gDG}\\
&g_{SD_L}=\frac{(-1)^k}{2}\left[L_k-\left(\frac{k L_{k-1}+(k+1) L_{k+1}}{2 k+1}\right)\right]=\frac{(-1)^k}{2}(1-x) L_k\\
&g_{SD_R}=\frac{1}{2}\left[L_k+\left(\frac{k L_{k-1}+(k+1) L_{k+1}}{2 k+1}\right)\right]=\frac{1}{2}(1+x) L_k\\
&g_{HU_L}=\frac{(-1)^k}{2}\left[L_k-\left(\frac{(k+1) L_{k-1}+k L_{k+1}}{2 k+1}\right)\right]\nonumber\\ &\text{ and }
g_{HU_R}=\frac{1}{2}\left[L_k+\left(\frac{(k+1) L_{k-1}+k L_{k+1}}{2 k+1}\right)\right]\label{eq:gHU}\\
&\text{where}\nonumber\\
&L_i\text{: $i$-th order Legendre polynomial.} \nonumber
\end{align}\vspace{-\baselineskip}

Later, Vincent et al.~\cite{vincent2011new} generalized Huynh's work to obtain an infinite family of correction functions that satisfy energy stability:
\begin{align}
&g_{L}=\frac{(-1)^{k}}{2}\left[L_{k}-\left(\frac{\eta_{k} L_{k-1}+L_{k+1}}{1+\eta_{k}}\right)\right]
\text{ and }
g_{R}=\frac{1}{2}\left[L_{k}+\left(\frac{\eta_{k} L_{k-1}+L_{k+1}}{1+\eta_{k}}\right)\right],\label{eq:correction_functions_general_form}\\
&\text{where }\eta_{k}=\frac{c(2 k+1)\left(a_{k} k !\right)^{2}}{2},\text{ and }a_{k}=\frac{(2 k) !}{2^{k}(k !)^{2}}.\nonumber
\end{align}\vspace{-\baselineskip}

They were able to recover previous schemes, with the following $c$ parameters:
$c_{DG}=0$, $c_{SD}=\frac{2 k}{(2 k+1)(k+1)\left(a_k k !\right)^2}$, $c_{HU}=\frac{2(k+1)}{(2 k+1) k\left(a_k k !\right)^2}$.

The final stage consists of using the derivative of the total flux, $\hat{f}^{\delta}$ at each quadrature point to obtain the semi-discrete scheme:
\begin{equation}
    \partial_t\hat{u}_i^\delta+\sum\limits_{m=0}^{k}\hat{f}_m^{\delta D}d_rl_m(r_i)+(\hat{f}_L^{\delta I}-\hat{f}_L^{\delta D})d_rg_L(r_i)+(\hat{f}_R^{\delta I}-\hat{f}_R^{\delta D})d_rg_R(r_i)=0.
    \label{eq:semidiscrete_ESFR}
\end{equation}

\subsection{Bilinear Form}

To follow the outlines of Cockburn's proof on optimal error estimates for the discontinuous Galerkin method~\cite{cockburn1999discontinuous}, a bilinear form representing ESFR schemes has to be derived.
In this subsection, a bilinear form of the ESFR method applied to the one-dimensional linear advection is formulated. This derivation is similar to the approach used by Vincent et al. to build the Sobolev-type energy norm for the ESFR scheme~\cite{vincent2011new}, but it differs since they did not use a test function to obtain the norm.

The following proof is restricted to the linear case, therefore the flux introduced in the conservation law (\ref{eq:pb}) is $f=au$, where $a$ denotes the advection velocity. Under such an assumption, the semi-discrete scheme (\ref{eq:semidiscrete_ESFR}) can be written as:
\begin{equation}
\partial_t\hat{u}_i^\delta+\hat{a}\sum\limits_{m=0}^{k}\hat{u}_m^{\delta}d_rl_m(r_i)+(\hat{f}_L^{\delta I}-\hat{a}\hat{u}_L^{\delta})d_rg_L(r_i)+(\hat{f}_R^{\delta I}-\hat{a}\hat{u}_R^{\delta})d_rg_R(r_i)=0,
    \label{eq:semidiscrete_ESFR_linear}
\end{equation}
where $\hat{a}=\frac{2a}{\Delta_j}$, $\hat{u}_L^{\delta}=\hat{u}^{\delta}(-1,t)$ and $\hat{u}_R^{\delta}=\hat{u}^{\delta}(1,t)$. By multiplying (\ref{eq:semidiscrete_ESFR_linear}) by the Lagrange polynomial $l_i$ and summing over $i$, one obtains:
\begin{equation}
\partial_t\hat{u}^\delta+\hat{a}\partial_r\hat{u}^{\delta}+(\hat{f}_L^{\delta I}-\hat{a}\hat{u}_L^{\delta})d_rg_L+(\hat{f}_R^{\delta I}-\hat{a}\hat{u}_R^{\delta})d_rg_R=0.
    \label{eq:ESFR_linear}
\end{equation}

The bilinear form differs from that acquired for the discontinuous Galerkin method since it contains both the contribution of the solution and it's $k^{th}$ derivative. Let us multiply (\ref{eq:ESFR_linear}) by any test function $\hat{v}_h\in \mathcal{P}^k(\Omega_s)$, such that $v_h(x,t)=\hat{v}_h(\Gamma(x),t)$ and $v_h(t)\in V_h=V_h^k=\{v\in\mathcal{L}^1(\Omega):v|_{I_j}\in \mathcal{P}^k(I_j),j=1,...,N\}\hspace{3mm} \forall t\in[0,T]$, and integrate over $\Omega_s$
\begin{align}
\int_{-1}^1\hat{v}_h\partial_t\hat{u}^\delta dr+\hat{a}\int_{-1}^1\hat{v}_h\partial_r\hat{u}^{\delta}dr&+(\hat{f}_L^{\delta I}-\hat{a}\hat{u}_L^{\delta})\int_{-1}^1\hat{v}_hd_rg_Ldr\nonumber\\
&+(\hat{f}_R^{\delta I}-\hat{a}\hat{u}_R^{\delta})\int_{-1}^1\hat{v}_hd_rg_Rdr=0.
    \label{eq:ESFR_linear_int}
\end{align}

The final two terms are integrated by parts and through the use of the boundary conditions on the correction functions (\ref{eq:condgl}) and (\ref{eq:condgr}):
\begin{align}
    \int_{-1}^1\hat{v}_h\partial_t\hat{u}^\delta dr+\hat{a}\int_{-1}^1\hat{v}_h\partial_r\hat{u}^{\delta}dr&+(\hat{f}_L^{\delta I}-\hat{a}\hat{u}_L^{\delta})\left(-\hat{v}_{hL}-\int_{-1}^1\partial_r\hat{v}_hg_Ldr\right)\nonumber\\&+(\hat{f}_R^{\delta I}-\hat{a}\hat{u}_R^{\delta})\left(\hat{v}_{hR}-\int_{-1}^1\partial_r\hat{v}_hg_Rdr\right)=0.
    \label{eq:ESFR_linear_int_byparts}
\end{align}
Now differentiating (\ref{eq:ESFR_linear}) $k$ times in space, one obtains:
\begin{equation}
    \partial_t\left(\partial^{k}_r\hat{u}^\delta\right)+{\hat{a}\partial_r^{k+1}\hat{u}^\delta}+(\hat{f}_L^{\delta I}-\hat{a}\hat{u}_L^{\delta})d^{k+1}_rg_L+(\hat{f}_R^{\delta I}-\hat{a}\hat{u}_R^{\delta})d^{k+1}_rg_R=0.
    \label{eq:ESFR_linear_kderivative}
\end{equation}
Then multiplying by $\partial^k_r\hat{v}_h$ and integrating over $\Omega_s$, yields:
\begin{align}
    \int_{-1}^1\partial^k_r\hat{v}_h\partial_t\left(\partial^{k}_r\hat{u}^\delta\right)dr+\hat{a}\int_{-1}^1\partial^k_r\hat{v}_h\partial_r^{k+1}\hat{u}^\delta dr&+(\hat{f}_L^{\delta I}-\hat{a}\hat{u}_L^{\delta})\partial^k_r\hat{v}_hd^{k+1}_rg_L\int_{-1}^1dr\nonumber\\&+(\hat{f}_R^{\delta I}-\hat{a}\hat{u}_R^{\delta})\partial^k_r\hat{v}_hd^{k+1}_rg_R\int_{-1}^1dr=0.
    \label{eq:ESFR_linear_kderivative_int}
\end{align}

Now, both contributions are added together with a factor two for the solution (\ref{eq:ESFR_linear_int_byparts}) and a weight parameter $c$ for the $k^{th}$ derivative contribution (\ref{eq:ESFR_linear_kderivative_int}), one obtains:
\begin{align}
    2\int_{-1}^1\hat{v}_h\partial_t\hat{u}^\delta dr+c\int_{-1}^1\partial^k_r\hat{v}_h\partial_t\left(\partial^{k}_r\hat{u}^\delta\right)dr+2\hat{a}\int_{-1}^1\hat{v}_h\partial_r\hat{u}^{\delta}dr+c\hat{a}\int_{-1}^1\partial^k_r\hat{v}_h\partial_r^{k+1}\hat{u}^\delta dr&\nonumber\\
    +2(\hat{f}_L^{\delta I}-\hat{a}\hat{u}_L^{\delta})\left(-\hat{v}_{hL}-\int_{-1}^1\partial_r\hat{v}_hg_Ldr+c\partial^k_r\hat{v}_hd^{k+1}_rg_L\right)&\nonumber\\
    +2(\hat{f}_R^{\delta I}-\hat{a}\hat{u}_R^{\delta})\left(\hat{v}_{hR}-\int_{-1}^1\partial_r\hat{v}_hg_Rdr+c\partial^k_r\hat{v}_hd^{k+1}_rg_R\right)=0&.
    \label{eq:raw_sum}
\end{align}\vspace{-\baselineskip}

To ensure stability, the general form of the correction functions (\ref{eq:correction_functions_general_form}) were obtained as solutions to the following equations~\cite[Eq.(3.29) and Eq.(3.30)]{vincent2011new}:
\begin{align}
&\int_{-1}^1 g_L d_rl_i dr-c\left(d^k_r l_i\right)\left(d^{k+1}_r g_L\right)=0 \quad \forall i, \label{eq:assumption1}\\
&\int_{-1}^1 g_R d_rl_i dr-c\left(d^k_r l_i\right)\left(d^{k+1}_r g_R\right)=0 \quad \forall i,
\label{eq:assumption2}
\end{align}\vspace{-\baselineskip}

In addition, $\hat{v}_h$ being in $\mathcal{P}^k(\Omega_s)$, it can be expressed in Lagrange polynomial basis as \mbox{$\hat{v}_h=\sum\limits_{i=0}^k \hat{v_h}_il_i$} and thus the conditions (\ref{eq:assumption1}) and (\ref{eq:assumption2}) can be use to simplify (\ref{eq:raw_sum}):
\begin{align}
    &2\int_{-1}^1\hat{v}_h\partial_t\hat{u}^\delta dr+c\int_{-1}^1\partial^k_r\hat{v}_h\partial_t\left(\partial^{k}_r\hat{u}^\delta\right)dr+2\hat{a}\int_{-1}^1\hat{v}_h\partial_r\hat{u}^{\delta}dr
    \\ &+c\hat{a}\int_{-1}^1\partial^k_r\hat{v}_h\partial_r^{k+1}\hat{u}^\delta dr+2(\hat{f}_L^{\delta I}-\hat{a}\hat{u}_L^{\delta})\left(-\hat{v}_{hL}\right)+2(\hat{f}_R^{\delta I}-\hat{a}\hat{u}_R^{\delta})\left(\hat{v}_{hR}\right)=0.\nonumber
    \label{eq:sum_ref_space}
\end{align}\vspace{-\baselineskip}
By mapping back (\ref{eq:sum_ref_space}) to the physical element domain $I_j$, one obtains:
\begin{align}
    &2\int_{I_j}v_h\partial_tu^\delta dx+c\int_{I_j}\partial^k_xv_h\partial_t\left(\partial^{k}_xu^\delta\right)\left(\frac{\Delta_j}{2}\right)^{2k}dx&\nonumber\\
    &+2a\int_{I_j}v_h\partial_xu^{\delta}dx+ca\int_{I_j}\partial^k_xv_h\partial_x^{k+1}u^\delta \left(\frac{\Delta_j}{2}\right)^{2k}dx\nonumber\\
    &+2(f_{j-1/2}^{\delta I}-au_{j-1/2}^{\delta +})\left(-v_{h_{j-1/2}}^+\right)+2(f_{j+1/2}^{\delta I}-au_{j+1/2}^{\delta -})v_{h_{j+1/2}}^-=0,
    \label{eq:sum}
\end{align}
where superscripts $-$ and $+$ denote the left and right values at the interface.
In addition, we use the fully upwind numerical flux, as Cockburn did for DG:
\begin{equation}
    f^I(u,v)=a\frac{u+v}{2}-\frac{|a|}{2}(v-u).
    \label{eq:flux}
\end{equation}
Without lose of generality, let us consider $a>0$, then the last two terms in (\ref{eq:sum}) become:
\begin{align*}
    f_L^{\delta I}-au_L^{\delta +}&=\left\{a\frac{u^{\delta +}+u^{\delta -}}{2}-\frac{a}{2}(u^{\delta +}-u^{\delta -})-au^{\delta +}\right\}_{(x_{j-1/2})}\\
    &=-a\left(u^{\delta +}-u^{\delta -}\right)_{(x_{j-1/2})}\\
    &=-a[\![u^\delta]\!]_{j-1/2},
\end{align*}
and \qquad\qquad$f_R^{\delta I}-au^{\delta -}_R=0,$

where the symbol $[\![\cdot]\!]$ denotes the jump at the interface.

Finally, by summing over the entire physical domain and integrating over time, we acquire:
\begin{equation}
    \forall v_h/\forall t\in [0,T], \forall v_h(x,t) \in V_h \hspace{1cm} \mathcal{B}_{h_{FR}}(u^\delta,v_h)=0,
    \label{eq:bilinear_form_0}
\end{equation}
with the bilinear form finally defined as:
\begin{definition}{Bilinear form}
\begin{align}
     \mathcal{B}_{h_{FR}}(u^\delta,v_h)=&2\int_0^T\int_{\Omega}v_h\partial_tu^\delta dxdt+c\int_0^T\sum\limits_{j=1}^N\int_{I_j}\partial^k_xv_h\partial_t\left(\partial^{k}_xu^\delta\right)\left(\frac{\Delta_j}{2}\right)^{2k}dx\;dt\nonumber\\
     &+ca\int_0^T\sum\limits_{j=1}^N\int_{I_j}\partial^k_xv_h\partial_x^{k+1}u^\delta \left(\frac{\Delta_j}{2}\right)^{2k}dx\;dt\nonumber\\
     &+2a\int_0^T\sum\limits_{j=1}^N\int_{I_j}v_h\partial_xu^{\delta}dxdt+2a\int_0^T\sum\limits_{j=1}^N[\![u^\delta]\!]_{j-1/2}\left(v_{h_{j-1/2}}^+\right)dt
     \label{eq:bilinear_form}
\end{align}\vspace{-\baselineskip}
\end{definition}

\section{Derivation of the Error Estimate}
\label{chap:error_estimate}

This chapter aims to provide an error estimate for the ESFR scheme applied to the linear advection problem. It will be shown that the order of convergence depends on the ESFR parameter introduced above. The error estimate takes the form of Theorem~\ref{th:error_estimate}. The derivation of the proof can be found in section~\ref{subsec:Proof}; it follows the outlines of Cockburn's proof for the discontinuous Galerkin method~\cite[\S 2.7]{cockburn1999discontinuous}. 

\begin{theorem}{$\mathcal{L}^2$-Error estimate of ESFR schemes}
    \begin{align}
        \|u(T)-u^\delta(&T)\|_{L^{2}(\Omega)}\leq\|R_h(e(0))\|_{k,2}
        +\Biggl(|a|T(2\zeta_a+4\zeta_b\sqrt{C_k})(1+\Lambda_k)\left|u_{0}\right|_{H^{k+2}(\Omega)} \nonumber\\
        &\left.+C_{k,0}|u_0|_{\mathcal{H}^{k+1}(\Omega)}
         +|c|C_{k,k}\left(\frac{1}{2}\right)^k\sqrt{C_1\dots C_k}|u_0|_{\mathcal{H}^{k+1}(\Omega)}|a|\frac{T}{2}\right)(\Delta x)^{k+1}\nonumber\\
        &+\left(|c|\left(\frac{1}{2}\right)^k\sqrt{C_1\dots C_k}|u_0|_{\mathcal{H}^{k+1}(\Omega)}|a|\frac{T}{2}\right)(\Delta x)^k.
        \label{eq:error_estimate}
    \end{align}\vspace{-\baselineskip}
    \label{th:error_estimate}
\end{theorem}

\subsection{Preliminary}

Although we are looking for an $\mathcal{L}^2$-error estimate we will first use the broken Sobolev type norm that has been employed to demonstrate the stability of the ESFR scheme. This norm was defined in~\cite[Eq.(3.19)]{vincent2011new} and repeated here for completeness:
\begin{definition}{Broken Sobolev norm of the state}

The broken Sobolev norm of the state is defined as,
    \begin{equation}
    \|u^\delta\|_{k,2}=\left(\sum\limits_{j=1}^N\int_{I_j}(u^\delta)^2+\frac{c}{2}\left(\partial^{k}_xu^\delta\right)^2\left(\frac{\Delta_j}{2}\right)^{2k}dx\right)^{\frac{1}{2}},
    \label{eq:broken_sobolev_norm}
\end{equation}
it is a valid norm provided $c\in]c_-(k),+\infty[$, where $c_-(k)=\frac{-2}{(2k+1)(a_kk!)^2}$ and $a_k=\frac{(2k)!}{2^k(k!)^2}.$
\label{def:broken_sobolev_norm}
\end{definition}

From Definition~\ref{def:broken_sobolev_norm}, the energy stability can be expressed by the following equation~\cite[Eq.(3.30)]{vincent2011new}:
\begin{equation}
    \forall t\hspace{5mm}\partial_t\|u^\delta(t)\|_{k,2}\leq0.
    \label{eq:stability}
\end{equation}

Proving this result was the object of the original paper on ESFR~\cite{vincent2011new}. Subsequently, by integrating over time, we obtain the final expression for energy stability for all ESFR schemes:
\begin{equation}
    \forall t\geq0, \|u^\delta(t)\|_{k,2}\leq\|u^\delta(0)\|_{k,2}.
    \label{eq:stability_integrated}
\end{equation}
According to Vincent et al.~\cite{vincent2011new}, stability in the Sobolev norm implies stability in the $\mathcal{L}^2$-norm, since norms are equivalent in finite dimension, in particular, the Sobolev norm is greater than the $\mathcal{L}^2-$norm,
\begin{align*}
    \|u^\delta(T)\|^2_{k,2}&= \|u^\delta(T)\|^2_{\mathcal{L}^2(\Omega)}+\frac{c}{2}\sum\limits_{j=1}^N\int_{I_j}\left(\partial^{k}_xu^\delta(T)\right)^2\left(\frac{\Delta_j}{2}\right)^{2k}dx,\\
    &\geq\|u^\delta(T)\|^2_{\mathcal{L}^2(\Omega)}.
\end{align*}\vspace{-\baselineskip}
Hence, the $\mathcal{L}^2$-norm at $T$ is bounded by the broken Sobolev norm at $t=0$:
\begin{equation}
    \|u^\delta(T)\|_{\mathcal{L}^2(\Omega)}\leq\|u^\delta(0)\|_{k,2}
    \label{eq:L2_stability}
\end{equation}
Equation~(\ref{eq:L2_stability}) guarantees the stability of all ESFR schemes; however, to this date a theorem that establishes the formal proof for the error estimate has not been presented. The following subsection is the first of such an attempt based on the author's knowledge.

\subsection{Proof of Theorem~\ref{th:error_estimate}}
\label{subsec:Proof}

In order to demonstrate the error estimate we need to find an expression for the norm of the error. To do so we will use the previously defined bilinear form (\ref{eq:bilinear_form}). First let us derive the error equation; by applying the bilinear form on both the approximate solution, $u^\delta$ and the exact solution, $u$ we have from the derivation of the bilinear form (\ref{eq:bilinear_form_0}):
\begin{equation*}
    \mathcal{B}_{h_{FR}}(u^\delta,v_h)=0 \hspace{5mm} \forall v_h:v_h(t)\in V_h \hspace{5mm}\forall t\in (0,T),\\
\end{equation*}
and for the exact solution, $u$,
\begin{align*}
    \mathcal{B}_{h_{FR}}(u,v_h)&=2\int_0^T\int_{\Omega}v_h(\partial_tu+a\partial_xu) dx\;dt\\
    &+c\int_0^T\sum\limits_{j=1}^N\int_{I_j}\partial^k_xv_h\partial_x^k(\partial_tu+a\partial_xu)\left(\frac{\Delta_j}{2}\right)^{2k}dx\;dt\\
    &+2a\int_0^T\sum\limits_{j=1}^N[\![u]\!]_{j-1/2}\left(v_{h_{j-1/2}}^+\right)dt \quad \forall v_h:v_h(t)\in V_h, \quad\forall t\in (0,T).
\end{align*}
By definition of the problem presented in (\ref{eq:pb}), $\partial_tu+a\partial_xu=0$ and for a smooth solution $u$ {\color{black}across the computational domain:
\begin{equation}
    f_L^{\delta I}-au_L^{+}=-a[\![u]\!]_{j-1/2}=0,
\end{equation}
where, $f_L^{\delta I}$ is the interface flux and hence $\mathcal{B}_{h_{FR}}(u,v_h)=0$.} If the error is defined as the difference between the exact and the approximate solution: 
\begin{equation}
     e=u-u^\delta,
     \label{eq:def_error}
\end{equation} 
and by bilinearity, we obtain the error equation,
\begin{equation}
        \forall v_h\in V_h, \hspace{5mm} {B}_{h_{FR}}(e,v_h)=0.
        \label{eq:error_equation}
\end{equation}

We now need to evaluate ${B}_{h_{FR}}(e,e)$ to obtain the norm of the error; however, since $e\notin V_h$ we instead need to consider {\color{black} evaluating the bilinear form of the projection of the error on $V_h$. We first, provide a definition of the projection of a function on the space $V_h$,} 
\begin{definition}{Projection on $V_h$}

Given a function $w\in \mathcal{L}^2(\Omega)$ that is continuous on each $I_j$, we define $R_h(w)$ as the only element of the finite element space $V_h$ such that~\cite[Eq.(2.34)]{cockburn1999discontinuous}:
\begin{align}
    &\forall j=1,...,N:\nonumber\\
    &R_h(w)(x_{j,l})-w(x_{j,l})=0, \hspace{5mm}l=0,...,k,
    \label{eq:Rh}
\end{align}
where $x_{j,l}$ are the Gauss-Radau quadrature points of the interval $I_j$~\cite[Eq.(2.35)]{cockburn1999discontinuous}. We take
\begin{equation}
    x_{j, k}= \begin{cases}x_{j+1 / 2} & \text { if } a>0 \\ x_{j-1 / 2} & \text { if } a<0\end{cases}
    \label{eq:GR}
\end{equation}
\label{def:rh}
\end{definition}

In order to obtain an expression for the norm and to establish the error estimate, we focus on the quantity: $\mathcal{B}_h(R_h(e),R_h(e))$, and by the definition of the bilinear form (\ref{eq:bilinear_form}), we have:
\begin{align*}
     \mathcal{B}_{h_{FR}}(R_h(e),R_h(e))=&2\int_0^T\int_{\Omega}R_h(e)\partial_tR_h(e) dx\;dt\nonumber\\
     &+c\int_0^T\sum\limits_{j=1}^N\int_{I_j}\partial^k_xR_h(e)\partial_t\left(\partial^{k}_xR_h(e)\right)\left(\frac{\Delta_j}{2}\right)^{2k}dx\;dt,\nonumber\\
     &+ca\int_0^T\sum\limits_{j=1}^N\int_{I_j}\partial^k_xR_h(e)\cancelto{0}{\partial_x^{k+1}R_h(e)}\hspace{3mm}\left(\frac{\Delta_j}{2}\right)^{2k}dx\;dt,\nonumber\\
     &+2a\int_0^T\sum\limits_{j=1}^N\int_{I_j}R_h(e)\partial_xR_h(e)dx\;dt\nonumber\\
     &+2a\int_0^T\sum\limits_{j=1}^N[\![R_h(e)]\!]_{j-1/2}\left(R_h(e)_{j-1/2}^+\right)dt.
\end{align*} \vspace{-\baselineskip}
Since $R_h(e)\in V_h$ then $\partial_x^{k+1}R_h(e)=0$. Recalling Definition~\ref{def:broken_sobolev_norm} of the Sobolev type norm, the expression for $\mathcal{B}_{h_{FR}}(R_h(e),R_h(e))$ can be further simplified:
\begin{align*}
     \mathcal{B}_{h_{FR}}&(R_h(e),R_h(e))=\int_0^T\partial_t\left(\sum\limits_{j=1}^N\int_{I_j}(R_h(e))^2+\frac{c}{2}\left(\partial^{k}_xR_h(e)\right)^2\left(\frac{\Delta_j}{2}\right)^{2k}dx\right)dt\\
     &+a\int_0^T\sum\limits_{j=1}^N\int_{I_j}\partial_x\left(R_h(e)\right)^2dx\;dt+2a\int_0^T\sum\limits_{j=1}^N[\![R_h(e)]\!]_{j-1/2}\left(R_h(e)_{j-1/2}^+\right)dt\\
     =&\|R_h(e(T))\|^2_{k,2}-\|R_h(e(0))\|^2_{k,2}\\
     &+a\int_0^T\sum\limits_{j=1}^N\left\{(R_h(e)^-)^2-(R_h(e)^+)^2+2(R_h(e)^+)^2-2R_h(e)^-R_h(e)^ +\right\}_{j-1/2}dt\\
     =&\|R_h(e(T))\|^2_{k,2}-\|R_h(e(0))\|^2_{k,2} +a\int_{0}^{T} \sum\limits_{j=1}^N[\![R_h(e)]\!]_{j-1/2}^{2}dt.
\end{align*}\vspace{-\baselineskip}

It results in norms of the projection of the error with the addition of a positive term:
\begin{equation}
    \mathcal{B}_{h_{FR}}(R_h(e),R_h(e))=\|R_h(e(T))\|^2_{k,2}-\|R_h(e(0))\|^2_{k,2}+2\Theta_T(R_h(e)),
    \label{eq:start}
\end{equation}
where : 
    $\forall f: \hspace{5mm}\Theta_{T}\left(f\right)=\frac{a}{2} \int_{0}^{T} \sum\limits_{j=1}^N[\![f(t)]\!]_{j-1 / 2}^{2} d t.$

Then it remains to estimate the left-hand-side of (\ref{eq:start}). By definition of the expression for the error, (\ref{eq:def_error}), and using the linearity of the projection, we have:
\begin{align}
    R_h(e)-e&=R_h(u-u^\delta)-(u-u^\delta)\nonumber\\
    &=R_h(u)-u^\delta-u+u^\delta\nonumber\\
    &=R_h(u)-u.\label{eq:expanded_Rhe}
\end{align}
{\color{black}We can exploit this result by subtracting and adding back the error $e$ from the first term of the bilinear form.} 
Then the error equation (\ref{eq:error_equation}) enables us to further simplify the left-hand-side expression from (\ref{eq:start}):
\begin{align*}
    \mathcal{B}_{h_{FR}}(R_h(e),R_h(e))&=\mathcal{B}_{h_{FR}}(R_h(e)-e,R_h(e))+\mathcal{B}_{h_{FR}}(e,R_h(e))\\
    &=\mathcal{B}_{h_{FR}}(R_h(e)-e,R_h(e))+0\\
    &=\mathcal{B}_{h_{FR}}(R_h(u)-u,R_h(e)).
\end{align*}\vspace{-\baselineskip}
This equality allows us to expand the left-hand-side of (\ref{eq:start}) using the definition of the bilinear form (\ref{eq:bilinear_form}) and integrating by parts its second term, we obtain: 
\begin{align}
\mathcal{B}_{h_{FR}}(R_{h}(u)&-u, v_{h}) =2\int_{0}^{T} \int_{\Omega}\left(R_{h}\left(\partial_{t} u\right)(x, t)-\partial_{t} u(x, t)\right) v_{h}(x, t) dx\;dt \nonumber\\
&+2a\int_0^T\sum\limits_{j=1}^N\left\{(R_h(u)-u)^-\left(v_{h}^-\right)\right\}_{j+1/2}-\left\{(R_h(u)-u)^+\left(v_{h}^+\right)\right\}_{j-1/2}dt\nonumber\\
&-2\int_{0}^{T} \sum\limits_{j=1}^N\int_{I_{j}} a\left(R_{h}(u)(x, t)-u(x, t)\right) \partial_{x} v_{h}(x, t) dx\;dt\nonumber\\
&+c\int_0^T\sum\limits_{j=1}^N\int_{I_j}\partial^k_xv_h\partial_t\left(\partial^{k}_x[R_h(u)-u]\right)\left(\frac{\Delta_j}{2}\right)^{2k}dx\;dt\nonumber\\
&+ca\int_0^T\sum\limits_{j=1}^N\int_{I_j}\partial^k_xv_h\partial^{k+1}_x[R_h(u)-u]\left(\frac{\Delta_j}{2}\right)^{2k}dx\;dt\nonumber\\
&+2a\int_0^T\sum\limits_{j=1}^N[\![R_h(u)-u]\!]_{j-1/2}\left(v_{h_{j-1/2}}^+\right)dt.
\label{eq:err_exact_soln}
\end{align}
Since $R_h(u) \in V_h$ then naturally $\partial^{k+1}_xR_h(u)=0$. By Definition~\ref{def:rh} and the use of Gauss-Radau quadrature points (\ref{eq:GR}); if we assume that $a>0$ without lost of generality, we have: $(R_h(u)-u)^-=0$. Gathering only the surface terms from (\ref{eq:err_exact_soln}) leads to:
\begin{align*}
    &2a\int_0^T\sum\limits_{j=1}^N[\![R_h(u)-u]\!]_{j-1/2}\left(v_{h_{j-1/2}}^+\right)+\left\{\cancelto{0}{(R_h(u)-u)^-}\left(v_{h}^-\right)\right\}_{j+1/2}\\
    &\hspace{5cm}-\left\{(R_h(u)-u)^+\left(v_{h}^+\right)\right\}_{j-1/2}dt\\
    &=2a\int_0^T\sum\limits_{j=1}^N\left\{(R_h(u)-u)^+\left(v_{h}^+\right)\right\}_{j-1/2}-\left\{\cancelto{0}{(R_h(u)-u)^-}\left(v_{h}^+\right)\right\}_{j-1/2}\\
    &\hspace{5cm}-\left\{(R_h(u)-u)^+\left(v_{h}^+\right)\right\}_{j-1/2}dt\\
    &=0.
\end{align*}\vspace{-\baselineskip}
It remains:
\begin{align}
\mathcal{B}_{h_{FR}}\left(R_{h}(u)-u, v_{h}\right) &=2\int_{0}^{T} \int_{\Omega}\left(R_{h}\left(\partial_{t} u\right)(x, t)-\partial_{t} u(x, t)\right) v_{h}(x, t) dx\;dt \nonumber\\
&-2\int_{0}^{T} \sum\limits_{j=1}^N\int_{I_{j}} a\left(R_{h}(u)(x, t)-u(x, t)\right) \partial_{x} v_{h}(x, t) dx\;dt\nonumber\\
&+c\int_0^T\sum\limits_{j=1}^N\int_{I_j}\partial^k_xv_h\partial_t\left(\partial^{k}_x[R_h(u)-u]\right)\left(\frac{\Delta_j}{2}\right)^{2k}dx\;dt\nonumber\\
&+ca\int_0^T\sum\limits_{j=1}^N\int_{I_j}\partial^k_xv_h\partial^{k+1}_xu\left(\frac{\Delta_j}{2}\right)^{2k}dx\;dt.
\label{eq:bh(rh)}
\end{align}\vspace{-\baselineskip}

Then, to bound the first two remaining terms, we employ the main theorem of Bramble and Hilbert~\cite[Th.2]{bramble1970estimation}. This theorem aims to bound a linear functional. By adapting the existing proof, we aim to obtain a more precise bound that is tailored to our specific problem. Note that the proof holds for any $\mathcal{L}^p$-norm, however, for simplicity, we employ only the $\mathcal{L}^2$-norm.

First, let us assume that for any $j\in[\![0,N]\!]$, the interval $I_j$ is a bounded subdomain of $\mathbb{R}$ which is Euclidean vector space of dimension one. Before employing Bramble and Hilbert's theorem~\cite[Th.2]{bramble1970estimation}, several other mathematical tools are necessary, namely the quotient space, a modified norm on the Sobolev space, and a Markov inequality. 

\begin{definition}(see \cite[\S 3, p.114]{bramble1970estimation}){ Quotient space}

    For any interval $I_j$, the quotient space of $\mathcal{H}^{k+1}(I_j)$ with respect to $\mathcal{P}^k(I_j)$  is denoted $Q=\mathcal{H}^{k+1}(I_j)/\mathcal{P}^k(I_j)$. The element of $Q$ are the equivalence classes $[f]$, where $[f]$ is the class containing $f$. The equivalence relation is defined by:
    \begin{equation*}
        \forall f,g\in \mathcal{H}^{k+1}(I_j)\hspace{5mm} f\sim g \iff f-g\in \mathcal{P}^k(I_j).
    \end{equation*}
    \label{def:quotient_space}
\end{definition}\vspace{-\baselineskip}

This quotient space is a normed space defined by the following norm:
\begin{definition}(see \cite[\S 3, p.114]{bramble1970estimation})
    The norm on $Q$ is given by:
    \begin{equation*}
        \|[f]\|_Q=\inf\limits_{g\in[f]}\|g\|_{k+1,p,I_j}=\inf\limits_{p\in \mathcal{P}^k(I_j)}\|f+p\|_{k+1,p,I_j},
    \end{equation*}
    \label{def:normQ}
\end{definition}\vspace{-\baselineskip}

where $\|\cdot\|_{k+1,p,I_j}$ is a norm on Sobolev space that depends on interval length, defined as:
\begin{definition}(see \cite[(2.1),(2.3), p.113]{bramble1970estimation}){ Modified Sobolev Norm}
    \begin{align}
        &\|f\|_{p,k,I_j}^p=\sum\limits_{l=0}^k(\Delta_j)^{lp}|f|_{p,l,I_j}^p\\
        &|f|_{p,l,I_j}=\sum\limits_{|\boldsymbol{\alpha}|=l}\|D^{\boldsymbol{\alpha}} f\|_{p,I_j} \label{eq:seminorm}\\
        &\|f\|_{p,I_j}=\left(\frac{1}{(\Delta_j)^{\mathrm{n}}}\int\limits_{I_j}{|f|^pdx}\right)^\frac{1}{p} = \frac{1}{(\Delta_j)^ \frac{\mathrm{n}}{p}}\|f\|_{\mathcal{L}^p(I_j)},
    \end{align}\vspace{-\baselineskip}
where $\mathrm{n}$ is the dimension and $\boldsymbol{\alpha}$ is a multi-index:
    \begin{equation*}
       \boldsymbol{\alpha}=(\alpha_1,...,\alpha_{\mathrm{n}}) \text{ and } |\boldsymbol{\alpha}|=\sum\limits_{i=0}^{\mathrm{n}}\alpha_i\text{, }D^\alpha=\left(\frac{\partial}{\partial x_1}\right)^{\alpha_1} \cdots\left(\frac{\partial}{\partial x_{\mathrm{n}}}\right)^{\alpha_{\mathrm{n}}}.
    \end{equation*}
    \label{def:Modified_Sobolev_norms}
\end{definition}\vspace{-\baselineskip}

The quotient space is related to this semi-norm (\ref{eq:seminorm}) through the following equivalence lemma:

\begin{lemma}(see \cite[Th.1]{bramble1970estimation})
    The norm on the quotient space is equivalent to the semi-norm on Sobolev space, more precisely:
    \begin{align*}
        &\exists \zeta \text{ independent of } \Delta_j \text{ and } f \text{ such that: }\\
        &\forall f \in \mathcal{H}^{k+1}(I_j) \hspace{5mm} (\Delta_j)^{k+1}|f|_{p,k+1,I_j}\leq\|[f]\|_Q\leq \zeta(\Delta_j)^{k+1}|f|_{p,k+1,I_j}.
    \end{align*}\vspace{-\baselineskip}
    \label{th:equivalent_norms}
\end{lemma}\vspace{-\baselineskip}
\begin{proof}
{\color{black}The proof can be found in appendix \ref{app:equivalent_morms}.}
\end{proof}

This lemma is required to prove the subsequent approximation theorem for linear functionals. Lastly, we restate the Markov inequality derived by Rivière et al.~\cite{ozisik2010constants} as it would be used within our rederivation of Bramble and Hilbert~\cite[Th.2]{bramble1970estimation}.
\begin{theorem}(see~\cite[Th.1]{ozisik2010constants}){Markov type inequality}
    \begin{equation}
        \forall p\in \mathcal{P}^k(I_j)\hspace{5mm}\left\|\partial_xp\right\|_{\mathcal{L}^2(I_j)} \leq \frac{2 \sqrt{C_k} }{\Delta_j}\|p\|_{\mathcal{L}^2(I_j)},
        \label{eq:Markov2}
    \end{equation}
    where $C_1=3, C_2=15, C_3=\frac{45+\sqrt{1605}}{2}$ and $C_4=\frac{105+3\sqrt{805}}{2}$.
    \label{th:markov}
\end{theorem}
\begin{proof}
    {\color{black}The proof of the second inequality Eq.(\ref{eq:Markov2}) is restated in appendix \ref{app:markov}.}
\end{proof}

Now that we have provided definitions for the quotient space, a modified norm on the Sobolev space, and the Markov inequality theorem, we finally state the primary theorem from Bramble and Hilbert~\cite[Th.2]{bramble1970estimation} in its original form. 
\begin{theorem}(see \cite[Th.2]{bramble1970estimation}){ Bramble Hilbert lemma on linear functional} 

    Let $F$ be a linear functional on $\mathcal{H}^{k+1}(I_j)$ which satisfies:
    \begin{align}
        (i)&\hspace{3mm}|F(f)|\leq C' \|f\|_{2,k+1,I_j} \hspace{3mm}\forall f\in \mathcal{H}^{k+1}(I_j)\text{ with } C' \text{independent of }\Delta_j \text{ and } f.\label{eq:Bh_first_condition}\\
        (ii)&\hspace{3mm}F(p)=0\hspace{3mm} \forall p\in \mathcal{P}^k(I_j)\label{eq:Bh_second_condition}\\
    \text{Then } &|F(f)|\leq C_1 (\Delta_j)^{k+1} |f|_{2,k+1,I_j},\forall f\in \mathcal{H}^{k+1}(I_j)\text{ and } C_1\text{independent of }\Delta_j \text{ and } f.\nonumber
    \end{align}
    \label{th:BH_linear_functional}
\end{theorem}

\begin{proof}
    the proof can be found in~\cite[Th.2]{bramble1970estimation} {\color{black}and is restated in appendix \ref{app:BH_linear_functional}.}
\end{proof}

{\color{black}As presented, Theorem~\ref{th:BH_linear_functional} provides bounds for linear functionals. For the theorem to be employed for our goal of establishing an error estimate for the ESFR scheme requires a rederivation of Theorem~\ref{th:BH_linear_functional} such as to ensure that the form can easily represent the terms in (\ref{eq:bh(rh)}). Lemma~\ref{lem:2first_terms} is an application of Theorem~\ref{th:BH_linear_functional} to the first two integrals of (\ref{eq:bh(rh)}).}

\begin{lemma}(see~\cite[Lem.22]{cockburn1999discontinuous}) if $w\in \mathcal{H}^{k+2}(I_j)$and $v_h\in \mathcal{P}^k(I_j)$ then
\begin{equation}
    \left|\int_{I_j}(R_h(w)-w)(x)v_h(x)dx\right|\leq c_k(\Delta x)^{k+1}|w|_{\mathcal{H}^{k+1}(I_j)}\|v_h\|_{\mathcal{L}^2(I_j)},
    \label{eq:first_term_bound}
\end{equation}
and
\begin{equation}
    \left|\int_{I_j}(R_h(w)-w)(x)\partial_xv_h(x)dx\right|\leq c_k(\Delta x)^{k+1}|w|_{\mathcal{H}^{k+2}(I_j)}\|v_h\|_{\mathcal{L}^2(I_j)}.
     \label{eq:second_term_bound}
\end{equation}
Both inequalities have a distinct constants $c_k$, which is solely dependent on k.
\label{lem:2first_terms}
\end{lemma}
\begin{proof}

First, let define a linear functional on $\mathcal{H}^{k+2}(I_j)$:
\begin{align}
    &\forall v_h \in \mathcal{P}^k(I_j), \forall w\in \mathcal{H}^{k+2}(I_j) \nonumber\\
    &F_{v_h}(w)=\frac{1}{\sqrt{\Delta_j}}\int_{I_j}(R_h(w)-w)(x)v_h(x)dx,
    \label{eq:functional_Gvh}
\end{align}
and the corresponding quotient space:
\begin{equation*}
    Q=\mathcal{H}^{k+1}(I_j)/P_k(I_j).
\end{equation*}

We obtain the first hypothesis $(i)$ of Theorem~\ref{th:BH_linear_functional} by employing the Cauchy-Schwartz inequality on $F_{v_h}(w)$. Then we can use the property of the Lebesgue constant to eliminate the norm of $R_h(u)$: $\|R_h(w)-w\|_{\mathcal{L}^2(I_j)}\leq (1+\Lambda_k)\|P_h(w)-w\|_{\mathcal{L}^2(I_j)}$; where the Lebesgue constant $\Lambda_k$ is the operator of the norm of $R_h$ with respect to the $\mathcal{L}^2$-norm: $\Lambda_k=\|R_h\|_{\mathcal{L}^2(I_j)}$ and only depends on the location of the quadrature points. $P_h(u)$ is defined as the best polynomial approximation of $u$ in $\mathcal{L}^2$ space~\cite[Eq.2.32]{cockburn1999discontinuous}. In other words, $P_h$ is the $\mathcal{L}^2$ orthogonal projection on $\mathcal{P}^k(I_j)$. 
At this juncture we can split the norm of the projection error with the triangle inequality and take advantage of the property of orthogonal projection: $\|P_h(w)\|_{\mathcal{L}^2(I_j)}\leq \|w\|_{\mathcal{L}^2(I_j)}$,
\begin{align}
    |F_{v_h}(w)|&\leq \frac{1}{\sqrt{\Delta_j}}\|R_h(w)-w\|_{\mathcal{L}^2(I_j)}\|v_h\|_{\mathcal{L}^2(I_j)}\nonumber\\
    &\leq \frac{1}{\sqrt{\Delta_j}}(1+\Lambda_k)\|P_h(w)-w\|_{\mathcal{L}^2(I_j)}\|v_h\|_{\mathcal{L}^2(I_j)}\nonumber\\
    &\leq \frac{1}{\sqrt{\Delta_j}}(1+\Lambda_k)\left(\|P_h(w)\|_{\mathcal{L}^2(I_j)}+\|w\|_{\mathcal{L}^2(I_j)}\right)\|v_h\|_{\mathcal{L}^2(I_j)}\nonumber\\
    &\leq \frac{1}{\sqrt{\Delta_j}}2(1+\Lambda_k)\|w\|_{\mathcal{L}^2(I_j)}\|v_h\|_{\mathcal{L}^2(I_j)}\nonumber\\
    &\leq \frac{1}{\sqrt{\Delta_j}}C"\|w\|_{\mathcal{L}^2(I_j)}\|v_h\|_{\mathcal{L}^2(I_j)}.\label{eq:Fvh_i}
\end{align}
We can then recast the constant as $C'=C"\|v_h\|_{\mathcal{L}^2(I_j)}$ and finally use Definition~\ref{def:Modified_Sobolev_norms} of norms to compare them and obtain the result in the modified Sobolev norm:
\begin{align*}
    {\color{black}|F_{v_h}(w)|}&\leq C'\|w\|_{2,I_j}\\
    &\leq C'\|w\|_{2,k+1,I_j}.
\end{align*}

Condition $(ii)$ of Theorem~\ref{th:BH_linear_functional} is easier to prove thanks to the advantageous property of Gauss-Radau quadrature points for error estimates,

\begin{proposition}(see \cite[Eq.(2.36)]{cockburn1999discontinuous}) Orthogonality Property for Gauss-Radau Quadrature
\begin{align}
&\forall \varphi \in P^{\ell}\left(I_{j}\right), \quad \ell \leq k, \quad \forall p \in P^{2 k-\ell}\left(I_{j}\right): \nonumber\\
&\int_{I_{j}}\left(R_{h}(p)(x)-p(x)\right) \varphi(x) d x=0. \label{eq:GRprop}
\end{align}
\label{prop:GRprop}
\end{proposition}
Since $v_h$ belongs to $\mathcal{P}^k(I_j)$, then we are able to show the second assumption (ii) of Theorem~\ref{th:BH_linear_functional}:
 \begin{equation}
     \forall p \in \mathcal{P}^k(I_j) \hspace{3mm} F_{v_h}(p)=0.
     \label{eq:Fvh_ii}
 \end{equation}
{\color{black}Since we established both assumptions (i) and (ii), we can proceed with the application of Theorem~\ref{th:BH_linear_functional}. Nevertheless, it is possible to derive an alternate expression for the constant $C_1$ that appears in this theorem by rederiving its proof for a particular case of $F_{v_h}.$}

First using the linearity of $F_{v_h}$ and (\ref{eq:Fvh_ii}), we have
\begin{equation*}
    |F_{v_h}(w)|=|F_{v_h}(w+p)|\hspace{3mm} \forall p \in \mathcal{P}^k(I_j).
\end{equation*}
Next with assumption (i) of Theorem~\ref{th:BH_linear_functional}, namely (\ref{eq:Fvh_i}), we get:
\begin{equation*}
    |F_{v_h}(w)|\leq C"\|w+p\|_{2,k+1,I_j}\|v_h\|_{\mathcal{L}^2(I_j)}.
\end{equation*}
Subsequently, taking $\inf\limits_{p\in \mathcal{P}^k(I_j)}$, we recover the norm in quotient space,
\begin{equation*}
    |F_{v_h}(w)|\leq C"\|[w]\|_Q\|v_h\|_{\mathcal{L}^2(I_j)}.
\end{equation*}
We can now apply Theorem~\ref{th:equivalent_norms}, to obtain a bound in terms of a modified Sobolev semi-norm, then by multiplying both sides by $\sqrt{\Delta_j}$ we get it in terms of the Sobolev semi-norm,
\begin{align*}
    |F_{v_h}(w)|&\leq \zeta_aC"(\Delta x)^{k+1}|w|_{2,k+1,I_j}\|v_h\|_{\mathcal{L}^2(I_j)}\\
    \left|\int_{I_j}(R_h(w)-w)(x)v_h(x)dx\right|&\leq \zeta_aC"(\Delta x)^{k+1}|w|_{\mathcal{H}^{k+1}(I_j)}\|v_h\|_{\mathcal{L}^2(I_j)}\\
    \left|\int_{I_j}(R_h(w)-w)(x)v_h(x)dx\right|&\leq 2\zeta_a(1+\Lambda_k)(\Delta x)^{k+1}|w|_{\mathcal{H}^{k+1}(I_j)}\|v_h\|_{\mathcal{L}^2(I_j)}.
\end{align*}
That completes the proof for the first inequality, (\ref{eq:first_term_bound}), of Lemma \ref{lem:2first_terms}. The procedure to bound the second inequality, (\ref{eq:second_term_bound}), is analogous. We obtain:

\begin{equation}
    \left|\int_{I_j}(R_h(w)-w)(x)\partial_xv_h(x)dx\right|\leq 2\zeta_b(1+\Lambda_k)(\Delta x)^{k+2}|w|_{\mathcal{H}^{k+2}(I_j)}\|\partial_xv_h\|_{\mathcal{L}^2(I_j)}.\label{eq:inequality_first_term}
\end{equation}

To bound the space derivative term, we can use the Markov type inequality (Theorem~\ref{th:markov}) for polynomials, (\ref{eq:inequality_first_term}) simplifies to:

\begin{align*}
    \left|\int_{I_j}(R_h(w)-w)(x)\partial_xv_h(x)dx\right|&\leq 4\zeta_b(1+\Lambda_k) \sqrt{C_k}(\Delta x)^{k+1}|w|_{\mathcal{H}^{k+2}(I_j)}\|v_h\|_{\mathcal{L}^2(I_j)}.
\end{align*}
\end{proof}

This completes the proof of Lemma~\ref{lem:2first_terms}. We can provide a bound for the first two terms of (\ref{eq:bh(rh)}) using Lemma~\ref{lem:2first_terms} with $w=\partial_tu$ for the first term and $w=u$ for the second term as well as $v_h=R_h(e)$ in both. We can now restate (\ref{eq:bh(rh)}) with the application of Lemma~\ref{lem:2first_terms}:

\begin{align}
     |\mathcal{B}_{h}(R_{h}(u)-u, R_h&(e))|
    \leq 4\zeta_a(1+\Lambda_k) (\Delta x)^{k+1}\int_{0}^{T} \sum_{1 \leq j \leq N}|\partial_{t} u|_{\mathcal{H}^{k+1}(I_j)}\-\|R_h(e)\|_{\mathcal{L}^2(I_j)} d t\nonumber\\
    &+ a8\zeta_b(1+\Lambda_k) \sqrt{C_k}(\Delta x)^{k+1}\int_{0}^{T} \sum_{1 \leq j \leq N}| u|_{\mathcal{H}^{k+2}(I_j)}\-\|R_h(e)\|_{\mathcal{L}^2(I_j)} d t\nonumber\\
    &+|c|\int_0^T\sum\limits_{j=1}^N\int_{I_j}|\partial^k_xR_h(e)\partial_t\left(\partial^{k}_x[R_h(u)-u]\right)|\left(\frac{\Delta_j}{2}\right)^{2k}dx\;dt\nonumber\\
    &+|ac|\int_0^T\sum\limits_{j=1}^N\int_{I_j}|\partial^k_xR_h(e)\partial^{k+1}_xu|\left(\frac{\Delta_j}{2}\right)^{2k}dx\;dt.
\label{eq:2first_terms_bounded}
\end{align}

For the first term in (\ref{eq:2first_terms_bounded}) we used: $|\partial_{t} u|_{\mathcal{H}^{k+1}(I_j)}=|a||u|_{\mathcal{H}^{k+2}(I_j)}$ by definition of the linear advection problem.
In addition, solving the linear advection equation with periodic boundary condition implies: $\forall t |u(t)|_{\mathcal{H}^{k+2},(I_j)}\leq|u_0|_{\mathcal{H}^{k+2}(I_j)}$, then:
\begin{align*}
|\mathcal{B}_{h}(R_{h}(u)&-u, R_h(e))| \\
\leq & 4\zeta_a(1+\Lambda_k) (\Delta x)^{k+1}\int_{0}^{T} \sum_{1 \leq j \leq N}|a|| u|_{\mathcal{H}^{k+2}(I_j)}\-\|R_h(e)\|_{\mathcal{L}^2(I_j)} d t\\
&+ a8\zeta_b(1+\Lambda_k) \sqrt{C_k}(\Delta x)^{k+1}\int_{0}^{T} \sum_{1 \leq j \leq N}| u|_{\mathcal{H}^{k+2}(I_j)}\-\|R_h(e)\|_{\mathcal{L}^2(I_j)} d t\\
&+|c|\int_0^T\sum\limits_{j=1}^N\int_{I_j}|\partial^k_xR_h(e)\partial_t\left(\partial^{k}_x[R_h(u)-u]\right)|\left(\frac{\Delta_j}{2}\right)^{2k}dx\;dt\\
&+|ac|\int_0^T\sum\limits_{j=1}^N\int_{I_j}|\partial^k_xR_h(e)\partial^{k+1}_xu|\left(\frac{\Delta_j}{2}\right)^{2k}dx\;dt,
\end{align*}

\noindent and by gathering the first two terms, we obtain a further simplified form:

\begin{align}
    \mathcal{B}_{h}(R_{h}(u)&-u, R_{h}(e))\nonumber\\ \leq& 2(2\zeta_a+4\zeta_b\sqrt{C_k})(1+\Lambda_k)(\Delta x)^{k+1}|a|\left|u_{0}\right|_{H^{k+2}(\Omega)} \int_{0}^{T}\left\|R_{h}(e(t))\right\|_{L^{2}(\Omega)} d t\nonumber\\ &+c\int_0^T\sum\limits_{j=1}^N\int_{I_j}\partial^k_xR_h(e)\partial_t\left(\partial^{k}_x[R_h(u)-u]\right)\left(\frac{\Delta_j}{2}\right)^{2k}dx\;dt\nonumber\\
    &+ac\int_0^T\sum\limits_{j=1}^N\int_{I_j}\partial^k_xR_h(e)\partial^{k+1}_xu\left(\frac{\Delta_j}{2}\right)^{2k}dx\;dt.
    \label{eq:factorisation_first_terms}
\end{align}

Now it remains to bound the last two terms in ~\ref{eq:factorisation_first_terms}. These two terms are specific to the ESFR scheme. For the penultimate term, we need an approximation result that can be found in Ciarlet's book~\cite[Th.3.1.5]{ciarlet2002finite}. A simplified version of this theorem for a one-dimensional element $I_j$ can be written as:

\begin{theorem}(see \cite[Th.3.1.5]{ciarlet2002finite}) {Error approximation in Sobolev space}

    For any polynomial interpolation $\Pi$ that leaves invariant the polynomials of degree $\leq k$ there exists a constant $C_{k,m}$ independent of $I_j$ such that:
    \begin{equation}
        \forall w \in \mathcal{H}^{k+1}(I_j),\hspace{5mm} |w-\Pi(w)|_{\mathcal{H}^{m}(I_j)}\leq C_{k,m} (\Delta_j)^{k+1-m}|w|_{\mathcal{H}^{k+1}(I_j)}, \hspace{5mm}0\leq m\leq k+1.
    \end{equation}
\label{th:approximation}
\end{theorem}

Let us establish an upper bound for the penultimate term of (\ref{eq:factorisation_first_terms}), by inverting the temporal and spatial derivatives and using the Cauchy-Schwartz inequality on\\ $\int_{I_j} |\partial^k_xv_h \partial_t \left(\partial^{k}_x[u-R_h(u)]\right)|dx$, we can 
{\color{black}then bound the term:} $\|\partial^k_xv_h\|_{\mathcal{L}^2(I_j)}$, with a Markov-type inequality, Theorem~\ref{th:markov}. The interpolation error term, $\|\partial^{k}_x[\partial_tu-R_h(\partial_tu)]\|_{\mathcal{L}^2(I_j)}$, can be bound using Theorem~\ref{th:approximation}, since $R_{h{|_{I_j}}}$ is the $k$-th order Lagrange interpolation. The Cauchy inequality, $ \left(\sum_{j=1}^N a_j b_j\right)^2 \leq\left(\sum_{j=1}^N a_j^2\right)\left(\sum_{j=1}^N b_j^2\right)$, allows us to reorganize the sum. We also use $|\partial_{t}u|_{\mathcal{H}^{k}(\Omega)}=|a||u|_{\mathcal{H}^{k+1}(\Omega)}$ as well as $\forall t|u|_{\mathcal{H}^{k+1}(\Omega)}\leq|u_0|_{\mathcal{H}^{k+1}(\Omega)}$ to simplify the temporal integral:
\begin{align}
    |c|\int_0^T&\sum\limits_{j=1}^N\int_{I_j}|\partial^k_xv_h\partial_t\left(\partial^{k}_x[u-R_h(u)]\right)|\left(\frac{\Delta_j}{2}\right)^{2k}dx\;dt\nonumber\\
    &\leq|c|\int_0^T\sum\limits_{j=1}^N\|\partial^k_xv_h\|_{\mathcal{L}^2(I_j)}\|\partial^{k}_x[\partial_tu-R_h(\partial_tu)]\|_{\mathcal{L}^2(I_j)}\left(\frac{\Delta_j}{2}\right)^{2k}dt\nonumber\\
    &\leq|c|\int_0^T\sum\limits_{j=1}^N\left(\frac{2}{\Delta_j}\right)^{k}\sqrt{C_1\dots C_k}\|v_h\|_{\mathcal{L}^2(I_j)}|\partial_tu-R_h(\partial_tu)|_{\mathcal{H}^{k}(I_j)}\left(\frac{\Delta_j}{2}\right)^{2k}dt\nonumber\\
    &\leq|c|\sqrt{C_1\dots C_k}\int_0^T\sum\limits_{j=1}^N\left(\frac{\Delta_j}{2}\right)^{k}\|v_h\|_{\mathcal{L}^2(I_j)}C_{k,k}\Delta_j^{k+1-k}|\partial_tu|_{\mathcal{H}^{k+1}(I_j)}dt\nonumber\\
    &\leq|c|C_{k,k}\left(\frac{1}{2}\right)^{k}\sqrt{C_1\dots C_k}\Delta x^{k+1}\int_0^T\|v_h\|_{\mathcal{L}^2(\Omega)}|\partial_tu|_{\mathcal{H}^{k+1}(\Omega)}dt\nonumber\\
    &\leq|c|C_{k,k}\left(\frac{1}{2}\right)^{k}\sqrt{C_1\dots C_k}\Delta x^{k+1}\int_0^T\|v_h\|_{\mathcal{L}^2(\Omega)}|a||u|_{\mathcal{H}^{k+2}(\Omega)}dt\nonumber\\
    &\leq|c|C_{k,k}\left(\frac{1}{2}\right)^{k}\sqrt{C_1\dots C_k}\Delta x^{k+1}|a||u_0|_{\mathcal{H}^{k+2}(\Omega)}\int_0^T\|v_h\|_{\mathcal{L}^2(\Omega)}dt,\label{eq:penultimate_term}
\end{align}

The procedure is similar for the last term of (\ref{eq:factorisation_first_terms}), except that there is no interpolation error to bound:
\begin{align}
    |ac|\int_0^T\sum\limits_{j=1}^N\int_{I_j}|\partial^k_xv_h\partial^{k+1}_xu|&\left(\frac{\Delta_j}{2}\right)^{2k}dx\;dt\nonumber\\
    \leq&|ac|\left(\frac{\Delta x}{2}\right)^{k}\sqrt{C_1\dots C_k}|u_0|_{\mathcal{H}^{k+1}(\Omega)}\int_0^T\|v_h\|_{\mathcal{L}^2(\Omega)}dt.\label{eq:Last_term}
\end{align}

To summarize, we obtained a bound for each term in (\ref{eq:factorisation_first_terms}), so we are able to estimate the left-hand-side of (\ref{eq:start}). By gathering all bounds (\ref{eq:factorisation_first_terms}), (\ref{eq:penultimate_term}) and (\ref{eq:Last_term}), (\ref{eq:start}) can be written as,
\begin{multline}
    \|R_h(e(T))\|^2_{k,2}+2\Theta_T(R_h(e)) \leq \|R_h(e(0))\|^2_{k,2}\\
    +\Biggl[\Biggl(2(2\zeta_a+4\zeta_b\sqrt{C_k})(1+\Lambda_k)\left|u_{0}\right|_{H^{k+2}(\Omega)}\\
    +\left.|c|C_{k,k}\left(\frac{1}{2}\right)^k\sqrt{C_1\dots C_k}|u_0|_{\mathcal{H}^{k+2}(\Omega)}\right)(\Delta x)^{k+1}\\
    \left. +\left(|c|\left(\frac{1}{2}\right)^k\sqrt{C_1\dots C_k}|u_0|_{\mathcal{H}^{k+1}(\Omega)}\right)(\Delta x)^k\right] |a|\int_{0}^{T}\left\|R_{h}(e(t))\right\|_{L^{2}(\Omega)}dt.   
\end{multline}
Since we have:
\begin{equation}
    \|R_h(e(T))\|^2_{k,2}\geq\left\|R_{h}(e(T))\right\|_{L^{2}(\Omega)}^2,
    \label{eq:l2_sobolev_comparison}
\end{equation} 
and $\Theta_T(R_h(e))\geq0$, it can be simplified as,
\begin{multline}
    \left\|R_{h}(e(T))\right\|_{L^{2}(\Omega)}^2 \leq \|R_h(e(0))\|^2_{k,2}\\
    +\Biggl[\Biggl(2(2\zeta_a+4\zeta_b\sqrt{ C_k})(1+\Lambda_k)\left|u_{0}\right|_{H^{k+2}(\Omega)}\\
    +\left.|c|C_{k,k}\left(\frac{1}{2}\right)^k\sqrt{C_1\dots C_k}|u_0|_{\mathcal{H}^{k+2}(\Omega)}\right)(\Delta x)^{k+1}\\
    \left. +\left(|c|\left(\frac{1}{2}\right)^k\sqrt{C_1\dots C_k}|u_0|_{\mathcal{H}^{k+1}(\Omega)}\right)(\Delta x)^k\right] |a|\int_{0}^{T}\left\|R_{h}(e(t))\right\|_{L^{2}(\Omega)}dt.  
\end{multline}

{\color{black}To finalize the proof we need to resolve the temporal integral term. Achieving this culminating step requires the use of a variation of Gronwall's Lemma. We restate it for completeness:}
\begin{theorem}(see \cite[Proposition 1.2]{barbu2016differential}) {Variation {\color{black}of the} Gronwall Lemma}

    Let $f:[0, T] \rightarrow \mathbb{R}$ be a continuous function which satisfies the following relation:
    \begin{equation}
        \frac{1}{2} f^2(t) \leq \frac{1}{2} f_0^2+\int_a^t \Psi(s) f(s) d s, \quad t \in[0, T],
    \label{eq:gronwall_lemma_condition}
    \end{equation}
    where $f_0 \in \mathbb{R}$ and $\Psi$ are non-negative continuous in $[0, T]$. Then the estimation
    \begin{equation}
        |f(t)| \leq\left|f_0\right|+\int_a^t \Psi(s) d s, \quad t \in[0, T]
        \label{eq:gronwall_lemma_result}
    \end{equation}
    holds.
    \label{th:Gronwall_lemma}
\end{theorem}
\begin{proof}
{\color{black} The proof is restated in appendix \ref{app:Gonwall_lemma}.}
\end{proof}

Hence applying Theorem~\ref{th:Gronwall_lemma} with $f(t)=\left\|R_{h}(e(t))\right\|_{L^{2}(\Omega)}$ continuous on $[0,T]$, $f_0=\|R_h(e(0))\|_{k,2}$ and 
\begin{align*}
 \Psi=&\frac{1}{2}\Biggl[\Biggl(2(2\zeta_a+4\zeta_b\sqrt{C_k})(1+\Lambda_k)\left|u_{0}\right|_{H^{k+2}(\Omega)}\\
    &\left.+|c|C_{k,k}\left(\frac{1}{2}\right)^k\sqrt{C_1\dots C_k}|u_0|_{\mathcal{H}^{k+2}(\Omega)}\right)(\Delta x)^{k+1}\nonumber\\
    &\left. +\left(|c|\left(\frac{1}{2}\right)^k\sqrt{C_1\dots C_k}|u_0|_{\mathcal{H}^{k+1}(\Omega)}\right)(\Delta x)^k\right] |a|,   
\end{align*}
both $f_0$ and $\Psi$ are positive constants on $[0,T]$, we obtain:
\begin{align}
    \left\|R_{h}(e(T))\right\|_{L^{2}(\Omega)} &\leq \|R_h(e(0))\|_{k,2}\nonumber\\
    &+\left[\Biggl(2(2\zeta_a+4\zeta_b\sqrt{C_k})(1+\Lambda_k)\left|u_{0}\right|_{H^{k+2}(\Omega)}\right. \nonumber\\
    &\hspace{2cm}+\left.|c|C_{k,k}\left(\frac{1}{2}\right)^k\sqrt{C_1\dots C_k}|u_0|_{\mathcal{H}^{k+2}(\Omega)}\right)(\Delta x)^{k+1}\nonumber\\
    &\left. +\left(|c|\left(\frac{1}{2}\right)^k\sqrt{C_1\dots C_k}|u_0|_{\mathcal{H}^{k+1}(\Omega)}\right)(\Delta x)^k\right] |a|\frac{T}{2}.
    \label{eq:after_Gronwall}
\end{align}

To complete the proof of Theorem~\ref{th:error_estimate}, recall that the objective is not to bound the projection of the error but the error itself. By using (\ref{eq:expanded_Rhe}), we can write the norm of the error as:
\begin{align*}
\|e(T)\|_{L^{2}(\Omega)} & =\left\|e(T)-R_{h}(e(T))+R_{h}(e(T))\right\|_{L^{2}(\Omega)}\\
&=\left\|u(T)-R_{h}(u(T))+R_{h}(e(T))\right\|_{L^{2}(\Omega)},
\end{align*}
then we can split the norm with the triangle inequality and bound the interpolation error term with Theorem~\ref{th:approximation},

\begin{align}
\|e(T)\|_{L^{2}(\Omega)}&\leq\left\|u(T)-R_{h}(u(T))\right\|_{L^{2}(\Omega)}+\left\|R_{h}(e(T))\right\|_{L^{2}(\Omega)} \\
&\leq C_{k,0}(\Delta x)^{k+1}|u(T)|_{\mathcal{H}^{k+1}(\Omega)}+\left\|R_{h}(e(T))\right\|_{L^{2}(\Omega)}
\label{eq:before_conclusion}
\end{align}

Utilizing once again: $|u(T)|_{\mathcal{H}^{k+1}(\Omega)}\leq|u_0|_{\mathcal{H}^{k+1}(\Omega)}$ as well as (\ref{eq:after_Gronwall}) in the previous inequality (\ref{eq:before_conclusion}), allows us to achieve the proof of Theorem~\ref{th:error_estimate}.

\subsection{Estimation of Constants}

There are no estimates for the constants $\zeta_A$ and $\zeta_B$, appearing in (\ref{eq:error_estimate}). There is also no simple analytical expression for the operator norm of the projection $\Lambda_k=\|R_h\|_{\mathcal{L}^2(I_j)}$. Only constants $C_{k,0}$ and $C_{k,k}$ can be estimated using Lemma~\ref{lem:constant_estimation}.


Arcangeli and Gout refined the interpolation error estimate of Theorem~\ref{th:approximation} in the particular case for Lagrange interpolation~\cite[Th.1-1]{arcangeli1976evaluation}. Using our notation, in one-dimension Theorem~\cite[Th.1-1]{arcangeli1976evaluation} can be expressed as,
\begin{lemma}(see \cite[Th.1-1]{arcangeli1976evaluation})
\begin{equation}
    \forall w \in \mathcal{H}^{k+1}(I_j),\hspace{5mm} |w-\Pi(w)|_{\mathcal{H}^{m}(I_j)}\leq C_{k,m} (\Delta_j)^{k+1-m}|w|_{\mathcal{H}^{k+1}(I_j)}, \hspace{5mm}0\leq m\leq k+1
\end{equation}
with
\begin{equation}
    C_{k,m}=\frac{1}{k !} \frac{1}{k+\frac{1}{2}} \sum_{i=1}^{k+1} \max\limits_{x \in I_j}\left(\left|\partial^m_xl_i(x)\right|\right)
\end{equation}
where $l_i$ are Lagrange basis polynomial of $\mathcal{P}^k(I_j)$
\label{lem:constant_estimation}
\end{lemma}
The $k$-th derivative simply becomes:
\begin{equation}
    \partial^k_xl_i(x)=\frac{k!}{\prod\limits^{k+1}_{q=1,q\neq i}(x_i-x_q)}=\text{constant}
\end{equation}
Then:
\begin{equation}
    C_{k,k}=\frac{1}{k+\frac{1}{2}} \sum_{i=1}^{k+1} \left(\left|\prod\limits^{k+1}_{q=1,q\neq i}(x_i-x_q)^{-1}\right|\right)
    \label{eq:Ckk}
\end{equation}

\section{Results}
\label{chap:results}

{\color{black}In this chapter, we present numerical results that validate the analytical error estimate derived in the previous chapter. Specifically, we test the estimate by solving the linear advection equation and comparing the numerical error with the established theoretical values. We compute the $\mathcal{L}^2$-error for different levels of mesh refinement and investigate the convergence rate of the numerical solutions.
Additionally, we explore the convergence properties of the method as it is employed to solve a more practical two-dimensional problem, namely the isentropic vortex flow governed by the Euler equations. For theses non-linear cases, the problems are carefully formulated to ensure the presence of known exact solutions in order to compute the $\mathcal{L}^2$-error. Additionally, we perform grid refinements to analyze the convergence order across a range of parameters $c$.}

\subsection{One-dimensional Linear Advection Equation}
\subsubsection{Problem}

The accuracy of the ESFR scheme with respect to the parameter $c$ is investigated by solving the one-dimensional linear advection equation:
\begin{align}
    &\partial_tu+a\partial_xu=0\text{ in }\Omega\times[0,T]\\
    &u(x,0)=\sin(x)\hspace{3mm}\forall x\in \Omega\\
    &\text{with periodic boundary conditions}\nonumber
\end{align}

The problem is defined in a domain $\Omega=[0, 2\pi]$ with periodic boundary conditions and an initial condition $u(x,0) = \sin(x)$, so the exact solution at the final time is simply: $u(x,T)=\sin(x-aT)$. For this simulation, the wave speed is set to $a=1$.

Spatial discretization is performed using the ESFR scheme implemented as a filtered discontinuous Galerkin method {\color{black}in strong form}, which is detailed in the paper by Allaneau et al.~\cite{allaneau2011connections} and Zwanenburg and Nadarajah~\cite{zwanenburg2016equivalence}.{\color{black} We use an upwind numerical flux as it is defined in (\ref{eq:flux}).} Numerical implementation for this work is based on a modification of the nodal DG code provided by Hesthaven and Warburton~\cite{hesthaven2007nodal}.
The examined polynomial orders span from $k=2$ to $k=5$.
The Legendre Gauss Lobatto (LGL) quadrature points are used. The range of the number of elements is adapted to polynomial order, ranging from 16 to 512.

The time evolution of the solution is computed using the low storage five-stage fourth-order explicit Runge-Kutta scheme (LSERK) with a final time of $T=\pi$. The details of this temporal discretization scheme are presented in~\cite{hesthaven2007nodal} and references therein. LSERK is commonly also referred to as RK45~\cite{vincent2011insights}. To ensure that the global error is dominated by the spatial discretization error, a sufficiently small time step is chosen.

\subsubsection{$\mathcal{L}^2$-error Results}

In order to quantify the accuracy of the ESFR error estimate for the one-dimensional linear advection equation, we compute the $\mathcal{L}^2$-error between the exact solution $u$ and the approximate solution $u^{\delta}$ at final time $T=\pi$. The exact $\mathcal{L}^2$-error is given by:
\begin{equation}
\|e\|_{\mathcal{L}^2(\Omega)}=\sqrt{\int_{\Omega} (u-u^{\delta})^2 dx}.
\end{equation}
To compute this integral, we use the integration quadrature rule with the same LGL quadrature points that were used to compute the approximate solution. The numerical $\mathcal{L}^2$-error is given by:
\begin{equation}
e^\delta=\sqrt{\sum_{j=1}^{N} \sum_{i=0}^{k} w_{i,j}(u(x_{i,j})-u^{\delta}(x_{i,j}))^2},
\label{eq:error_quadrature_rule}
\end{equation}
where $N$ is the number of elements, $k$ the polynomial order, $w_{i,j}$ and $x_{i,j}$ denote the quadrature weight and coordinate for the $i$-th LGL quadrature point in the $j$-th element, and $u(x_{i,j})$ and $u^{\delta}(x_{i,j})$ denote the exact and approximate solutions evaluated at $x_{i,j}$. 

{\color{black}Tables~\ref{tab:convergence_cDG} and ~\ref{tab:convergence_c+} and figure~\ref{fig:Convergence_linear_advection} show the $\mathcal{L}^2$-error of the solution against the number of elements for polynomial orders between 2 and 5 and for two ESFR schemes: parameterized by $c_{DG}=0$ and $c_+$ respectively. Where $c_+$ is the value of $c$ that allows the highest time step. This value depends on the degree of approximation and the temporal discretization of the scheme. It has been numerically determined by Vincent et al.~\cite{vincent2011insights} through a von Neumann analysis based on the stability requirement. In addition, tables~\ref{tab:convergence_cDG} and ~\ref{tab:convergence_c+} provide the $\mathcal{L}^\infty$-error of the solution as well as the odrer of accuracy (OOA) for both the $\mathcal{L}^2$-error and $\mathcal{L}^\infty$-error. According to the results, the expected order of accuracy is observed for each polynomial order and considered $c$ value.}

\begin{table}[htbp]
  \centering
  \sisetup{scientific-notation=true, round-mode=places, round-precision=2}
  \begin{adjustbox}{center}
  \begin{tabular}{cc||cccc}
     \multicolumn{2}{c||}{} &\multicolumn{4}{c}{$c_{DG}$}\\ 
     \hline
     $k$ & $N$ & $\mathcal{L}^2$-error & OOA & $\mathcal{L}^\infty$-error & OOA \\
    \hline
    \multirow{6}{5pt}{2}
    & 4 & \num{0.072745441} & - & \num{0.039983288} & -  \\
    & 8 & \num{0.010523548} & \num{2.789235689} & \num{0.007290794} & \num{2.455249203}  \\
    & 16 & \num{0.001334973} & \num{2.978738221} & \num{0.000985213} & \num{2.887568828}  \\
    & 32 & \num{0.00016751} & \num{2.994495477} & \num{0.000125428} & \num{2.973577495} \\
    & 64 & \num{2.10E-05} & \num{2.998662229} & \num{1.57E-05} & \num{2.993629658}  \\
    & 128 & \num{2.62E-06} & \num{2.999669969} & \num{1.97E-06} & \num{2.998434853}  \\
    \hline
    \multirow{6}{5pt}{3}
    & 4 & \num{0.005136294} & - & \num{0.021941387} & -\\
    & 8 & \num{0.000332137} & \num{3.950876128} & \num{0.001426368} & \num{3.943236224}  \\
    & 16 & \num{2.28052E-05} & \num{3.864342194} & \num{8.99836E-05} & \num{3.986541141} \\
    & 32 & \num{1.4258E-06} & \num{3.999524822} & \num{5.63816E-06} & \num{3.996365887}\\
    & 64 & \num{8.90992E-08} & \num{4.000213044} & \num{3.52597E-07} & \num{3.999130791}\\
    & 128 & \num{5.56924E-09} & \num{3.999859144} & \num{2.20406E-08} & \num{3.999787422} \\
    \hline
    \multirow{6}{5pt}{4}
    & 4 & \num{0.000345787} & - & \num{0.00122633} & -  \\
    & 8 & \num{1.11655E-05} & \num{4.952768479} & \num{3.95471E-05} & \num{4.954632745} \\
    & 16 & \num{3.53864E-07} & \num{4.979703502} & \num{1.24006E-06} & \num{4.995086002} \\
    & 32 & \num{1.12732E-08} & \num{4.972222684} & \num{3.8733E-08} & \num{5.000706663} \\
    & 64 & \num{3.61636E-10} & \num{4.962220252} & \num{1.21069E-09} & \num{4.999655224}\\
    & 128 & \num{1.14297E-11} & \num{4.983675766} & \num{3.78553E-11} & \num{4.999194826}\\
    \hline
    \multirow{6}{5pt}{5}
    & 2 & \num{0.001572988} & - & \num{0.004964984} & - \\
    & 4 & \num{2.13693E-05} & \num{6.20182329} & \num{6.53311E-05} & \num{6.247876045} \\
    & 8 & \num{3.46515E-07} & \num{5.946478497} & \num{1.04372E-06} & \num{5.967957827}\\
    & 16 & \num{5.45574E-09} & \num{5.988999155} & \num{1.63956E-08} & \num{5.99228762} \\
    & 32 & \num{8.54187E-11} & \num{5.997078807} & \num{2.56795E-10} & \num{5.996548371}\\
    & 64 & \num{1.33624E-12} & \num{4.477671722} & \num{4.02308E-12} & \num{5.996170938} \\
  \end{tabular}
    \end{adjustbox}
  \captionsetup{skip=0pt}
  \caption{1D linear advection, $\mathcal{L}^2$ and $\mathcal{L}^\infty$-convergence for polynomial order between 2 and 5 with $c_{DG}$}
  \label{tab:convergence_cDG}
\end{table}

\begin{table}[htbp]
  \centering
  \sisetup{scientific-notation=true, round-mode=places, round-precision=2}
  \begin{adjustbox}{center}
  \begin{tabular}{cc||cccc}
     \multicolumn{2}{c||}{} & \multicolumn{4}{c}{$c_+$} \\ 
     \hline
     $k$ & $N$ & $\mathcal{L}^2$-error & OOA & $\mathcal{L}^\infty$-error & OOA \\
    \hline
    \multirow{6}{5pt}{2}
    & 4  & \num{0.330836357} & - & \num{0.174221865} & - \\
    & 8  & \num{0.052405034} & \num{2.658340485} & \num{0.034149185} & \num{2.351002624} \\
    & 16 & \num{0.007007189} & \num{2.902797615} & \num{0.004962076} & \num{2.782835288} \\
    & 32 & \num{0.000897235} & \num{2.965277383} & \num{0.000657034} & \num{2.916903247} \\
    & 64 & \num{0.000113341} & \num{2.984811881} & \num{8.42E-05} & \num{2.96413887} \\
    & 128 & \num{1.42E-05} & \num{2.992851856} & \num{1.06E-05} & \num{2.983459987} \\
    \hline
    \multirow{6}{5pt}{3}
    & 4 & \num{0.005521328} & - & \num{0.02584185} & - \\
    & 8 & \num{0.00038193} & \num{3.853635558} & \num{0.001742279} & \num{3.890661602} \\
    & 16 & \num{2.74562E-05} & \num{3.798103447} & \num{0.000111263} & \num{3.968930956} \\
    & 32 & \num{1.76384E-06} & \num{3.960338507} & \num{6.99131E-06} & \num{3.992267192} \\
    & 64 & \num{1.10478E-07} & \num{3.996895137} & \num{4.37502E-07} & \num{3.99819974} \\
    & 128 & \num{6.91017E-09} & \num{3.998889682} & \num{2.73521E-08} & \num{3.9995655} \\
    \hline
    \multirow{6}{5pt}{4}
    & 4 & \num{0.000436003} & - & \num{0.001817129} & - \\
    & 8 & \num{1.79341E-05} & \num{4.603558103} & \num{6.53555E-05} & \num{4.797208506} \\
    & 16 & \num{5.90565E-07} & \num{4.92446883} & \num{2.059E-06} & \num{4.98829505} \\
    & 32 & \num{1.84859E-08} & \num{4.997598754} & \num{6.40175E-08} & \num{5.007331444} \\
    & 64 & \num{5.86511E-10} & \num{4.978124485} & \num{2.00247E-09} & \num{4.998613182} \\
    & 128 & \num{1.88689E-11} & \num{4.958076869} & \num{6.26645E-11} & \num{4.997988863} \\
    \hline
    \multirow{6}{5pt}{5}
    & 2 & \num{0.002423901} & - & \num{0.007433887} & - \\
    & 4 & \num{3.08164E-05} & \num{6.297487894} & \num{0.000114023} & 6.03 \\
    & 8 & \num{6.83561E-07} & \num{5.912404373} & \num{2.09684E-06} & 5.76 \\
    & 16 & \num{1.13492E-08} & \num{5.980109043} & \num{3.39111E-08} & 5.95 \\
    & 32 & \num{1.79794E-10} & \num{5.998900205} & \num{5.31255E-10} & 6.00\\
    & 64 & \num{2.81142E-12} & \num{6.009863629} & \num{8.24429E-12} & 6.01\normalsize\\
  \end{tabular}
    \end{adjustbox}
  \captionsetup{skip=0pt}
  \caption{1D linear advection, $\mathcal{L}^2$ and $\mathcal{L}^\infty$-convergence for polynomial order between 2 and 5 with ${c_+}$}
  \label{tab:convergence_c+}
\end{table}
\begin{figure}[ht]
    \centering
    \includegraphics[width=0.7\linewidth,trim=0 0 60 20,clip]{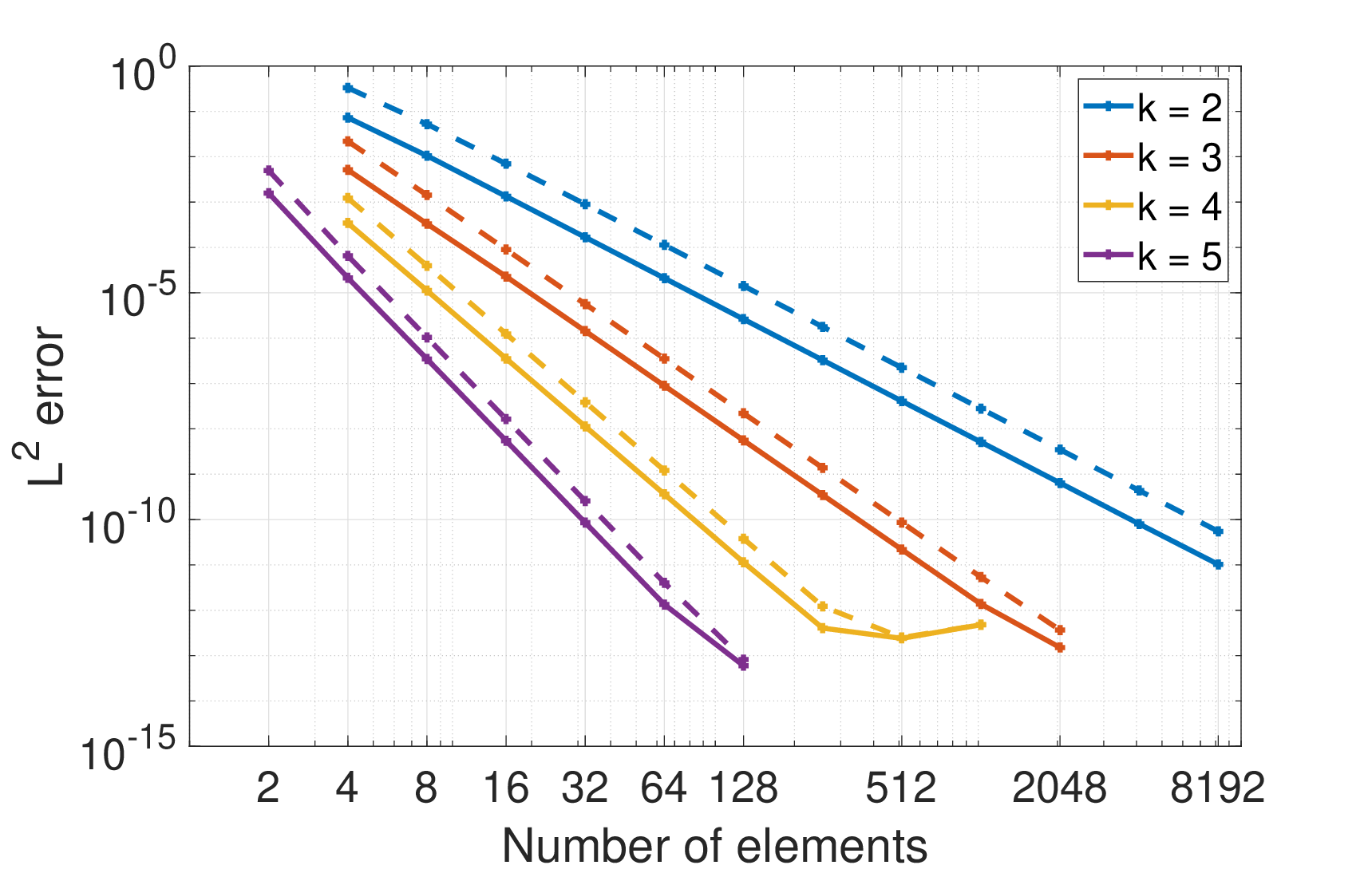}
    \captionsetup{skip=0pt}
    \caption{1D linear advection, convergence plot for different polynomial order, solid line $c_{DG}$, dash line $c_+$}
    \label{fig:Convergence_linear_advection}
\end{figure}
In addition, we compute the error estimate, Theorem~\ref{th:error_estimate}, as a function of $c$. We express the error estimate (\ref{eq:error_estimate}) with a simplified notation:
\begin{equation}
    \|e(T)\|_{L^{2}(\Omega)}\leq\lambda+(\mu+\nu |c|)\Delta x^{k+1}+\eta |c|\Delta x^k,
    \label{eq:simplifiederrorestimate}
\end{equation}
\begin{align*}
    \text{where: }&\lambda = \|R_h(e(0))\|_{k,2}, \\
    &\mu = |a|T(2\zeta_a+4\zeta_b\sqrt{C_k})(1+\Lambda_k)\left|u_{0}\right|_{H^{k+2}(\Omega)}+C_{k,0}|u_0|_{\mathcal{H}^{k+1}(\Omega)},\\
    &\nu = C_{k,k}\left(\frac{1}{2}\right)^k\sqrt{C_1\dots C_k}|u_0|_{\mathcal{H}^{k+1}(\Omega)}|a|\frac{T}{2},\\
    &\eta=\left(\frac{1}{2}\right)^k\sqrt{C_1\dots C_k}|u_0|_{\mathcal{H}^{k+1}(\Omega)}|a|\frac{T}{2},
\end{align*}
with $a=1$, $T=\pi$ and $\left|u_{0}\right|_{H^{k+2}(\Omega)}=\left|u_{0}\right|_{H^{k+1}(\Omega)}=\pi$, the value of $C_{k,k}$ is obtained using (\ref{eq:Ckk}). The $\lambda$ term is independent of $\Delta x$ and is small enough to be negligible. 

Since there is no analytical estimate for constants $\zeta_A$ and $\zeta_B$, $\mu$ is computed numerically, by setting $c$ to 0. The constant $\mu$ is then evaluated based on the computed numerical error for a specific grid size, $\Delta x$ and polynomial order, $k$ through the following relationship: $\mu={e^\delta}/{\Delta x^{k+1}}$. Once evaluated, the value, $\mu$ is kept frozen as we increase the value of parameter $c$.

{\color{black}Now that all constants are known either analytically or numerically, we are able to plot both the ($k+1$)- and $k$-th order terms as a function of $c$.
The results depicted in Figures~\ref{fig:error_p3_N32} and~\ref{fig:error_p3_N256} illustrate the error estimates for $k=3$ with 32 and 256 elements, respectively and Figures~\ref{fig:error_p5_N16} and~\ref{fig:error_p5_N64} provides visuals for $k=5$ with 16 and 64 elements. The curve for the ($k+1$)-th order term, $(\mu+\nu |c|)\Delta x^{k+1}$, is highlighted in green, while the $k$-th order term, $\eta |c|\Delta x^k$ is in blue.}

\begin{figure}[ht]
    \centering
    \begin{subfigure}{0.49\linewidth}
        \centering
        \includegraphics[width=\linewidth,trim=7 5 75 25,clip]{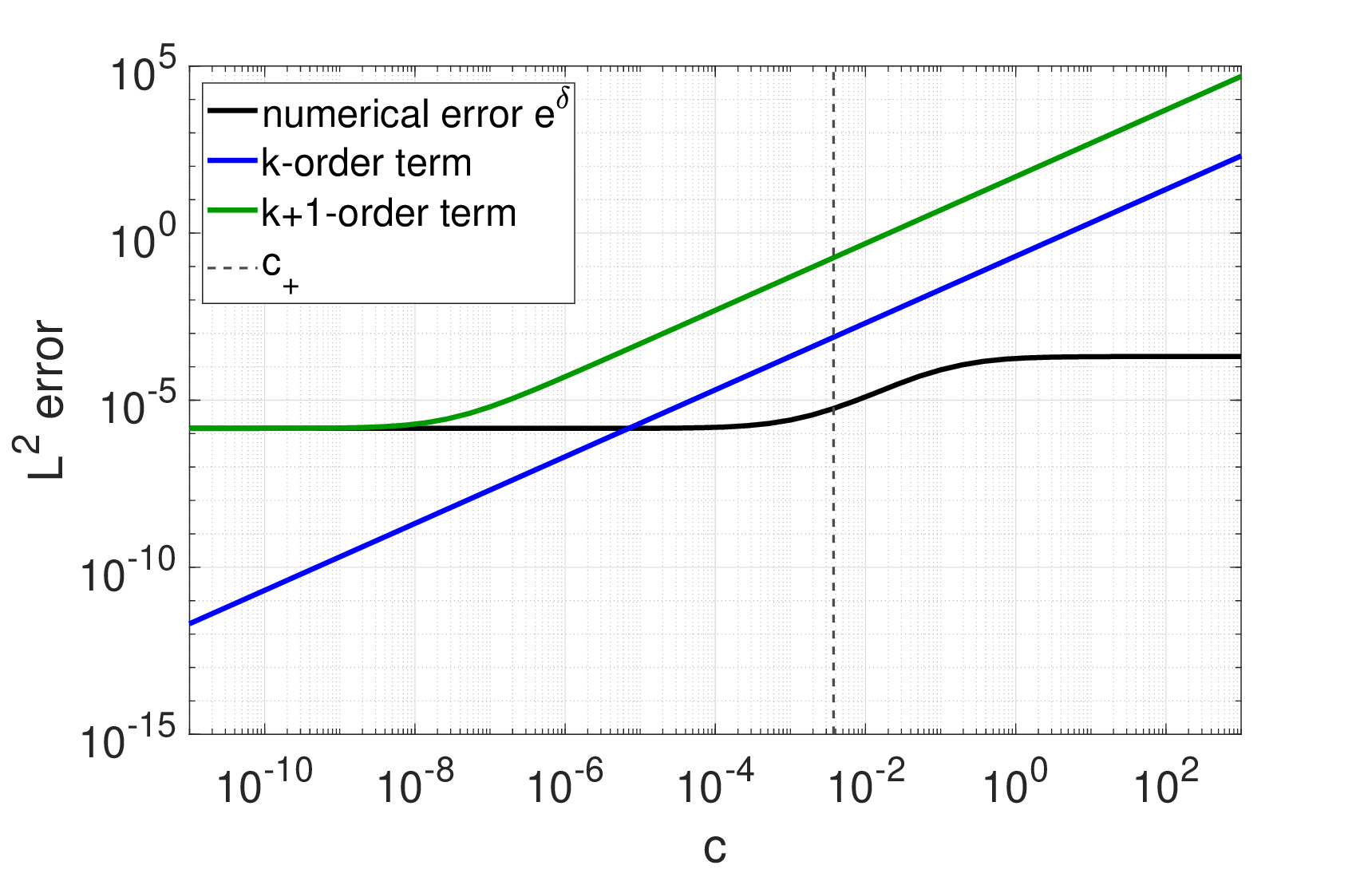}
        \captionsetup{skip=0pt}
        \caption{$N=32$ elements}
        \label{fig:error_p3_N32}
    \end{subfigure}
    \hfill
    \begin{subfigure}{0.49\linewidth}
        \centering
        \includegraphics[width=\linewidth,trim=7 5 75 25,clip]{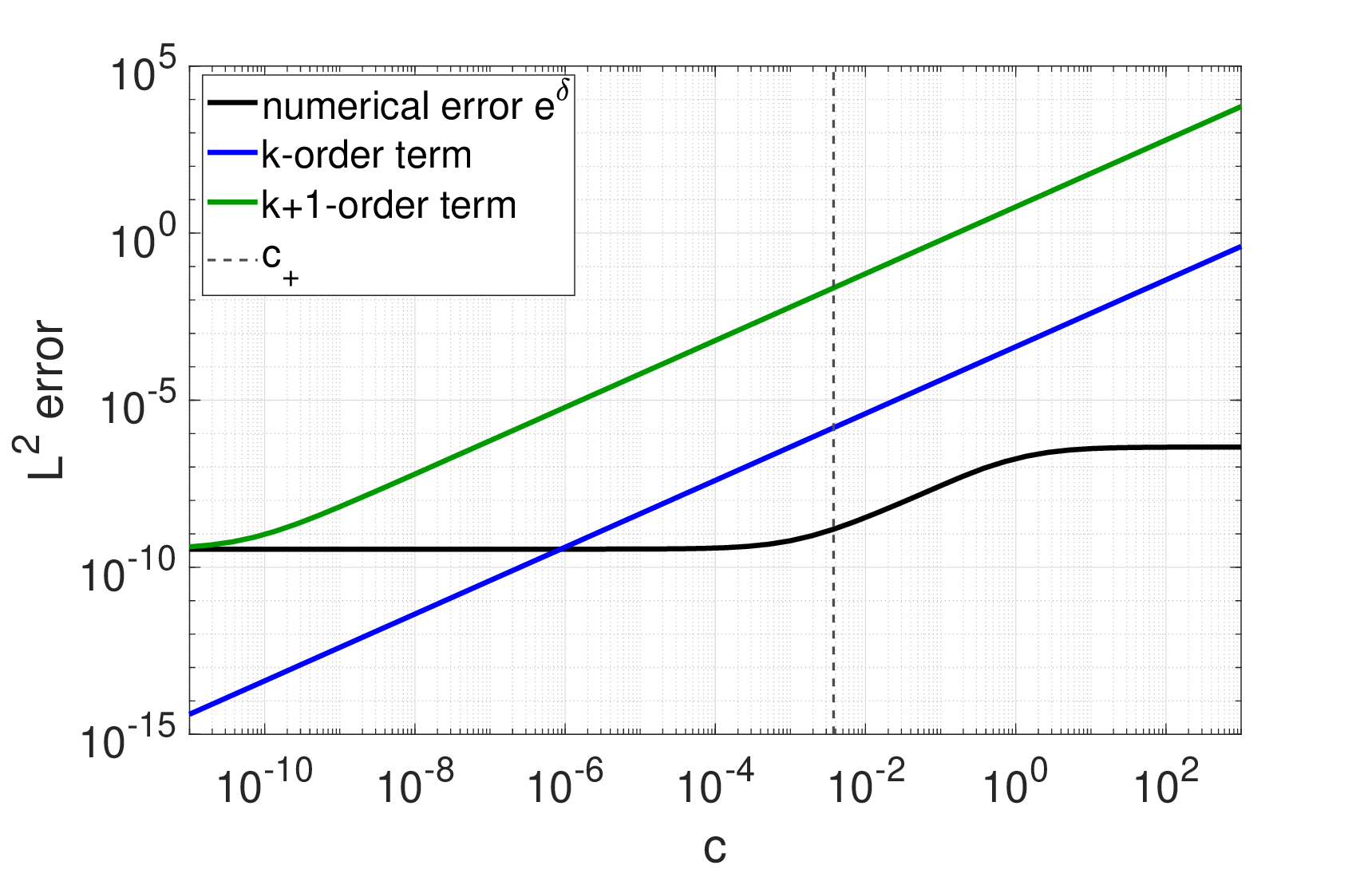}
        \captionsetup{skip=0pt}
        \caption{$N=256$ elements}
        \label{fig:error_p3_N256}
    \end{subfigure}
    \captionsetup{skip=-7pt}
    \caption{1D linear advection, numerical and estimated $\mathcal{L}^2$-error as function parameter $c$ between $10^{-11}$ and $10^3$, polynomial order $k=3$.}
    \label{fig:error_p3}
\end{figure}

\begin{figure}[ht]
    \centering
    \begin{subfigure}{0.49\linewidth}
        \centering
        \includegraphics[width=\linewidth,trim=7 5 58 25,clip]{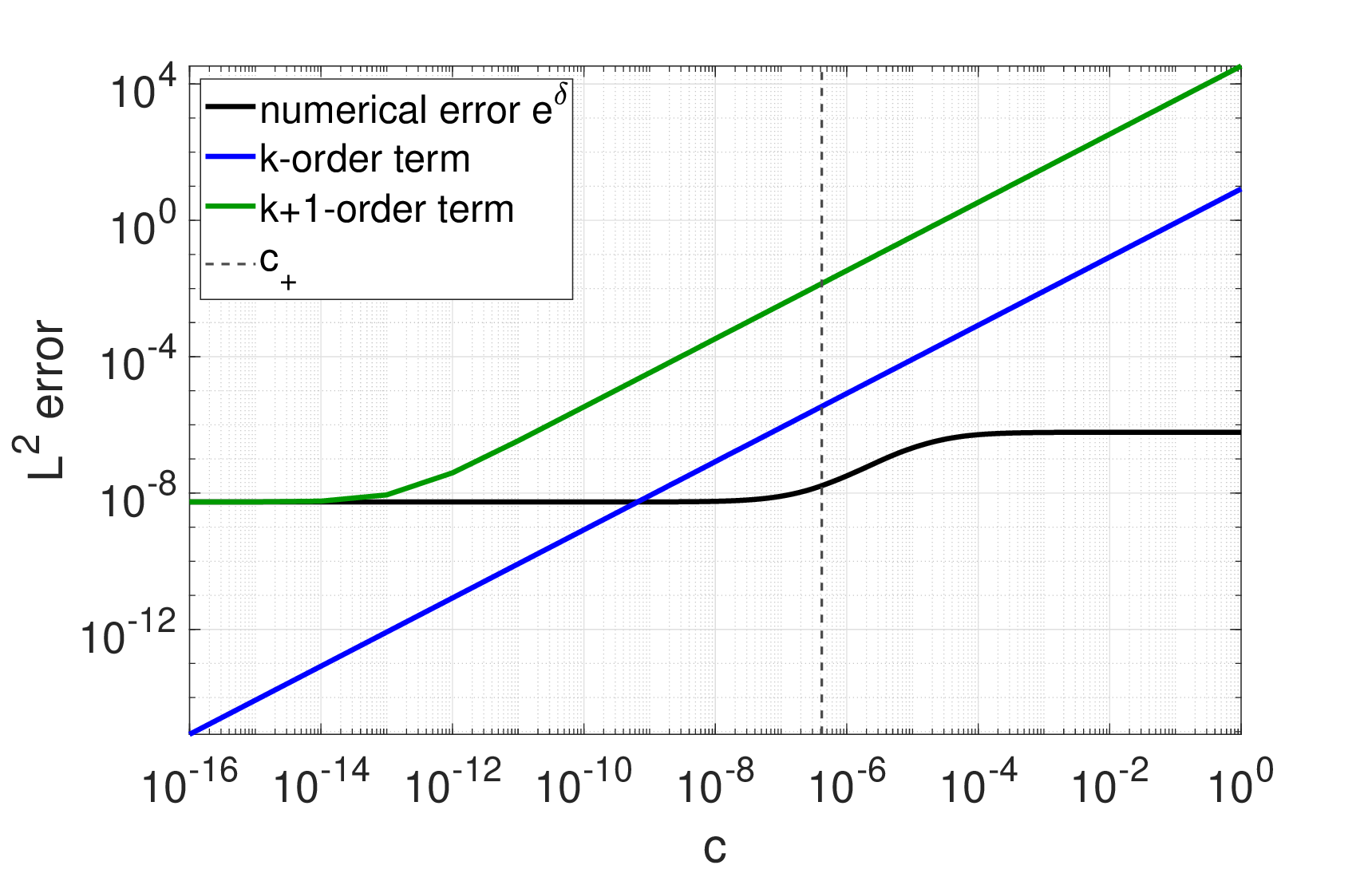}
        \captionsetup{skip=0pt}
        \caption{$N=16$ elements}
        \label{fig:error_p5_N16}
    \end{subfigure}
    \hfill
    \begin{subfigure}{0.49\linewidth}
        \centering
        \includegraphics[width=\linewidth,trim=7 5 58 25,clip]{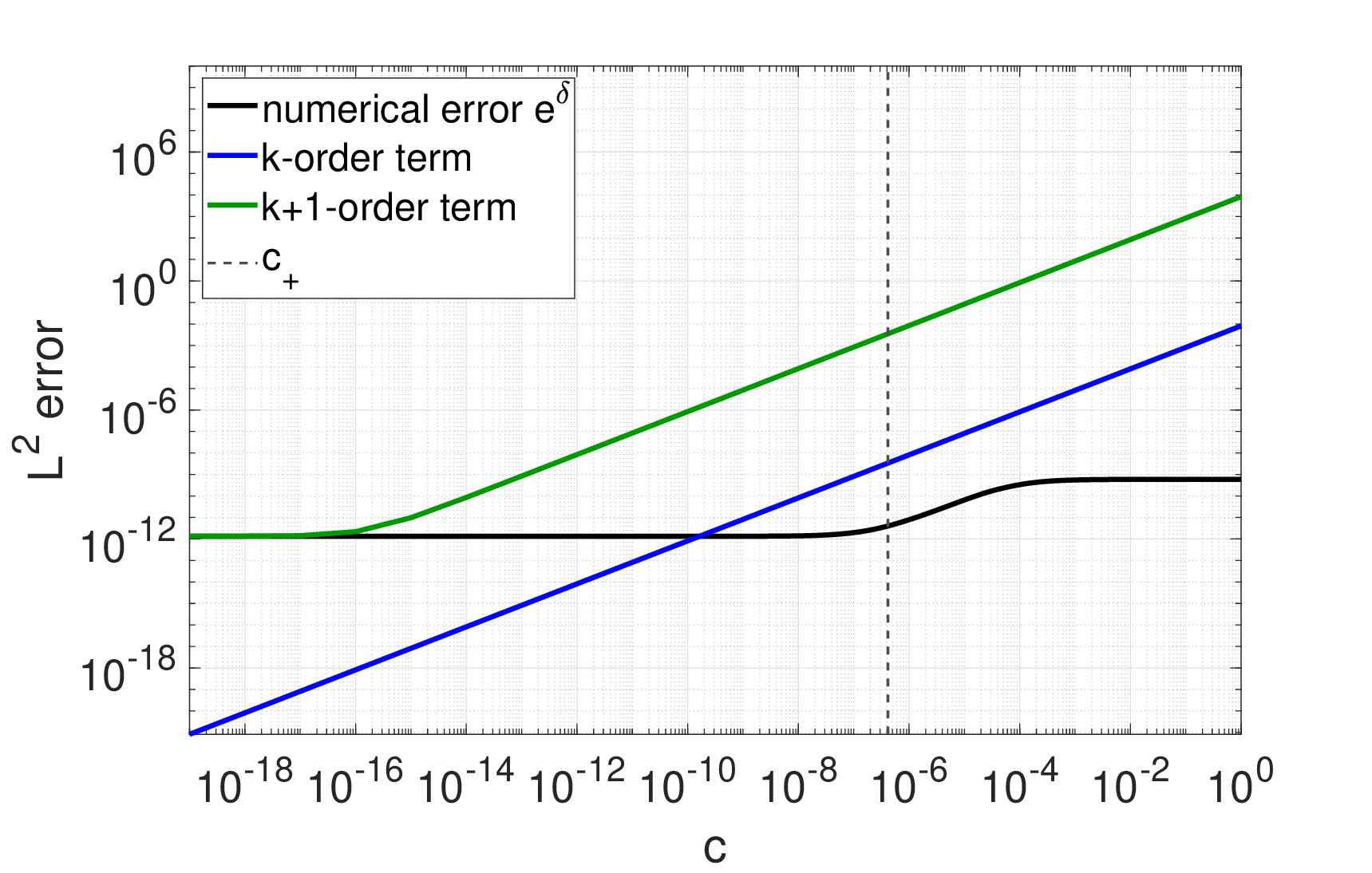}
        \captionsetup{skip=0pt}
        \caption{$N=64$ elements}
        \label{fig:error_p5_N64}
    \end{subfigure}
    \captionsetup{skip=-7pt}
    \caption{1D linear advection, numerical and estimated $\mathcal{L}^2$-error as function parameter $c$ between $10^{-16}$ and $1$, polynomial order $k=5$.}
    \label{fig:error_p5}
\end{figure}

In each case, by construction, the ($k+1$)-th order term (green) matches the numerical $\mathcal{L}^2$-error for very low values of the parameter $c$. The $\mathcal{L}^2$-error is always smaller than the estimate, but its behaviour differs for high values of $c$. The numerical $\mathcal{L}^2$-error $e^\delta$ exhibits almost no sensitivity to changes in $c$ for both low and high values of $c$ with a smooth increasing step in between, while the estimate is almost constant in $c$ only for low values of $c$ and increases linearly after. In other words, the term $\mu\Delta x^{k+1}$, which is independent of $c$, is dominant over $\nu |c|\Delta x^{k+1}+\eta |c|\Delta x^k$ only when $c$ is sufficiently small.

{\color{black}
\begin{remark}
    If we were able to find a sharper comparison between the $\mathcal{L}^2$-norm and Sobolev type norm of the projection of the error (\ref{eq:l2_sobolev_comparison}), we could obtain similar shapes of estimate and numerical error also at high values of $c$.
\end{remark}}

This study of the $\mathcal{L}^2$-error is too limited to draw general conclusions in regards to the sharpness of the estimate. First, we neglected the constant term $\lambda = \|R_h(e(0))\|_{k,2}$, and secondly, we performed a hybrid computation of the error estimation between analytical ($\nu,\eta$) and numerical ($\mu$) constants from (\ref{eq:simplifiederrorestimate}). Those numerically determined constants could be problem dependent and affect in either way the sharpness of the error estimate.

\subsubsection{Order of Accuracy Results}

In order to investigate the order of accuracy of the numerical method, we conducted a series of experiments with different polynomial orders and number of elements. The number of elements was chosen depending on the polynomial order, to ensure that the error remained higher than machine precision. The OOA was measured by computing the $\mathcal{L}^2$-error for each grid size and calculating the average slope of the $\mathcal{L}^2$-error versus the number of elements on a log scale.

Figure~\ref{fig:OOA_linear_advection} shows the order of convergence results for polynomial orders from 2 to 5. For each polynomial order, the ESFR scheme is ($k+1$)-th order accurate up to about $c_+$ and then drops down to $k$-th order accuracy for high values of $c$. Similar results were reported by Castonguay~\cite[Fig.3.6]{Castonguay}.

Since the decrease is monotonic, it is difficult to determine graphically the precise value of $c$ where the order is lost. {\color{black}Moreover, this value depends slightly on the 
grid size employed when computing the average order of accuracy.} As figure~\ref{fig:OOAk2_evolution} illustrates the lost of order occurs at lower values of $c$ for coarser grids.

\begin{figure}[ht]
    \centering
    \begin{subfigure}{0.49\linewidth}
    \centering
        \includegraphics[width=\linewidth,trim=20 0 50 27,clip]{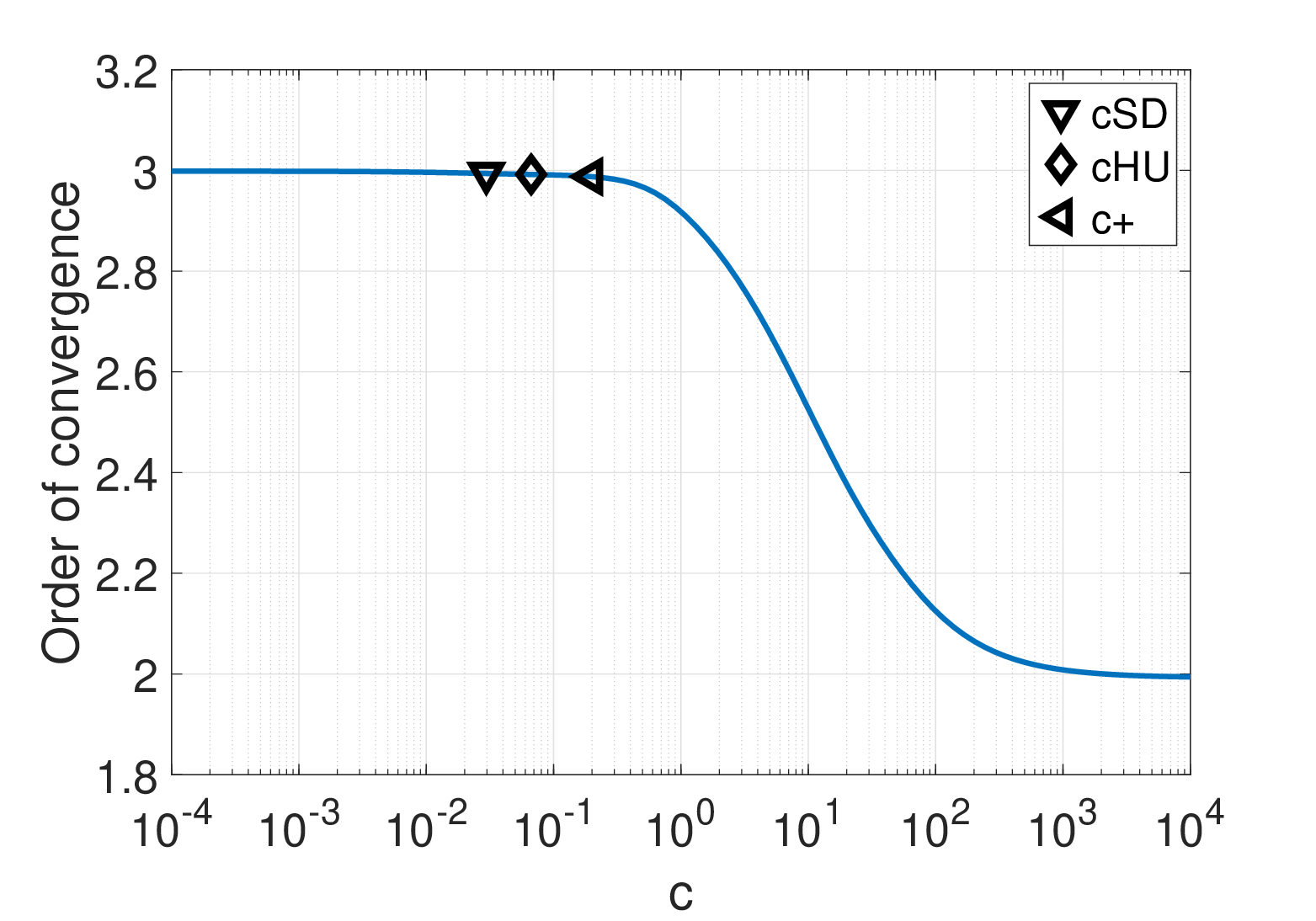}
        \captionsetup{skip=0pt}
        \caption{$k=2, N=16$ to $512$}
        \label{fig:OOAk2}
    \end{subfigure}
    \hfill
    \begin{subfigure}{0.49\linewidth}
    \centering
        \includegraphics[width=\linewidth,trim=20 0 50 27,clip]{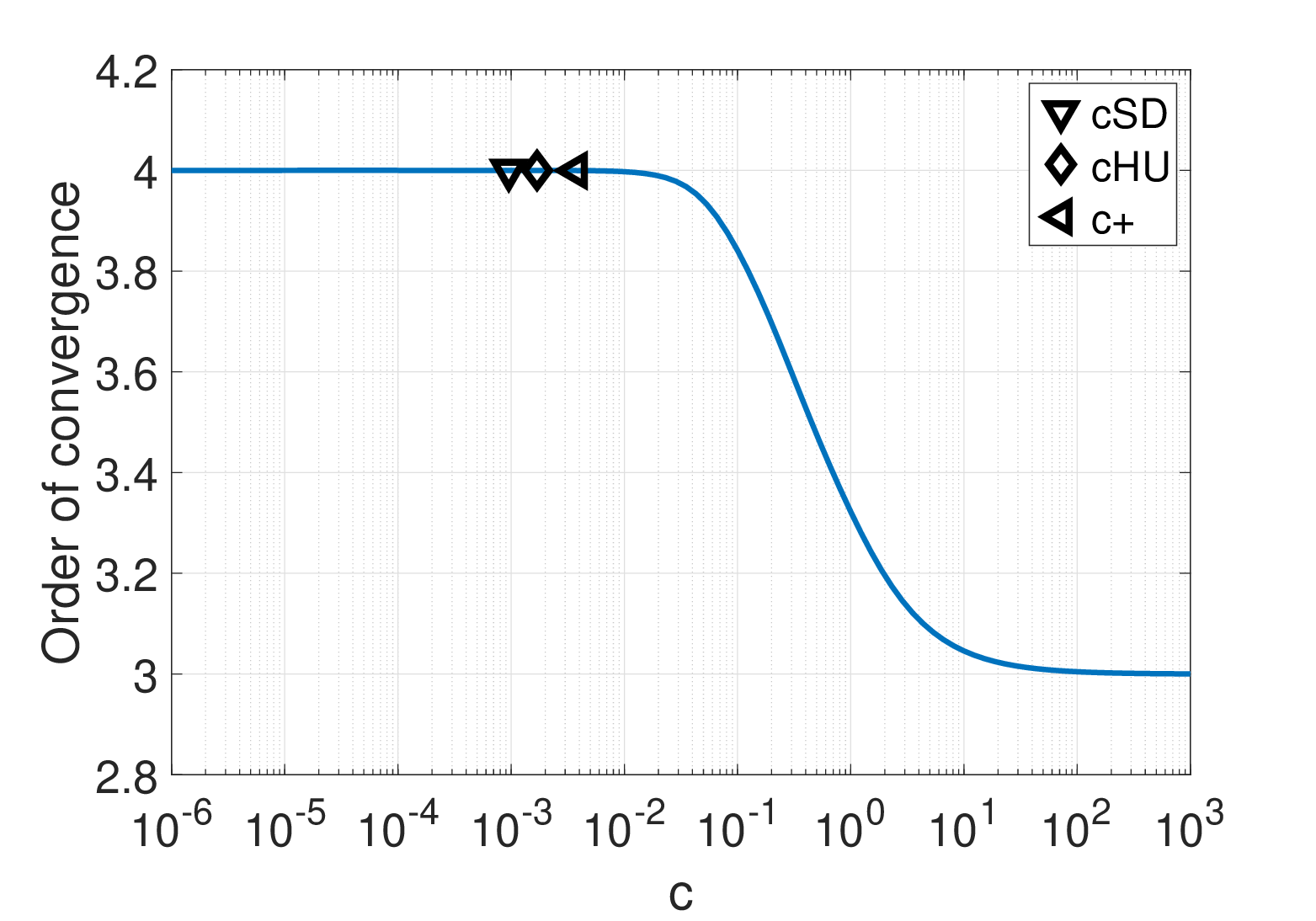}
        \captionsetup{skip=0pt}
        \caption{$k=3, N=32$ to $256$}
        \label{fig:OOAk3}
    \end{subfigure}
    
    \vspace{-3mm}
    
    \begin{subfigure}{0.48\linewidth}
    \centering
        \includegraphics[width=\linewidth,trim=20 0 50 27,clip]{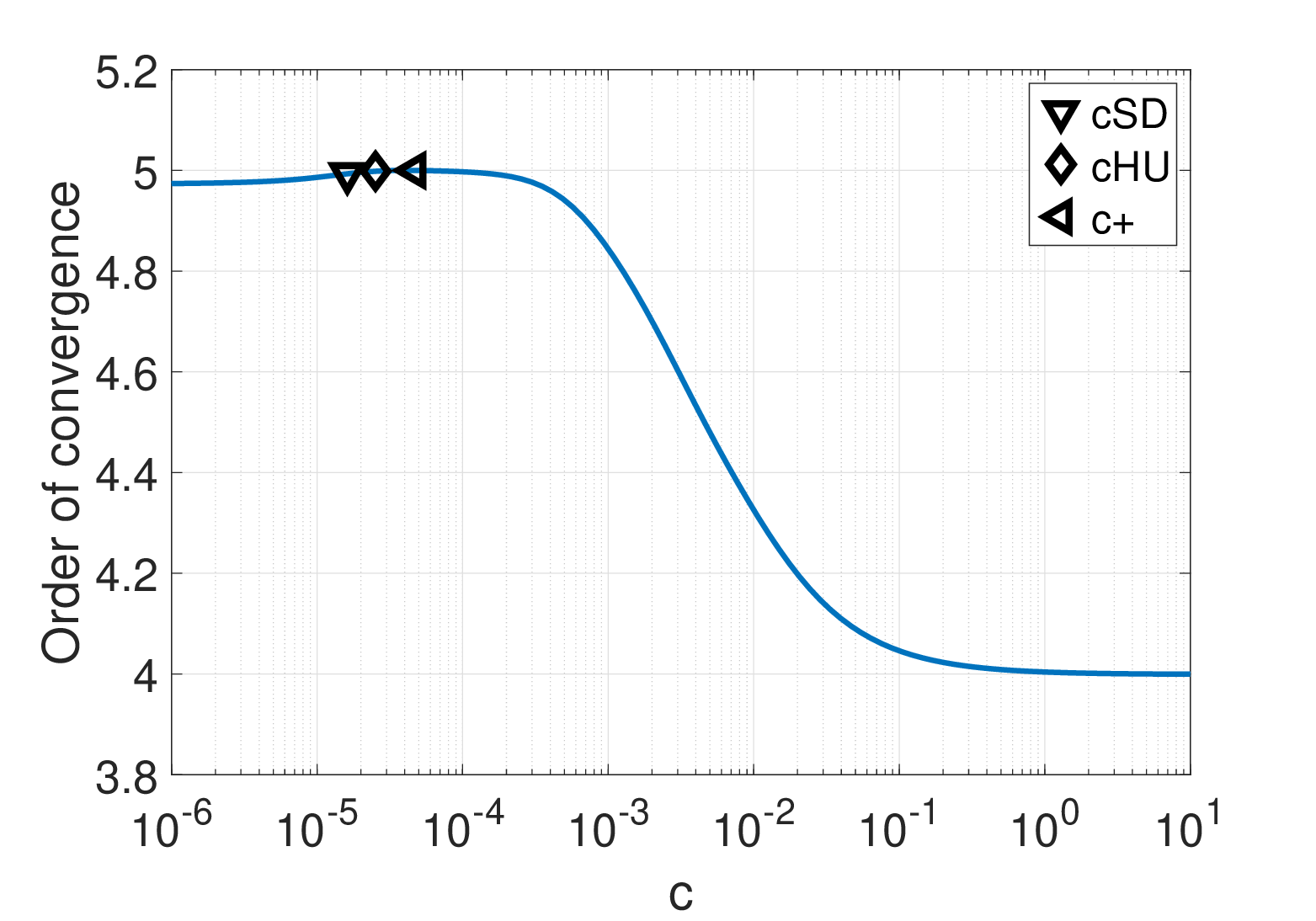}
        \captionsetup{skip=0pt}
        \caption{$k=4, N=16$ to $128$}
        \label{fig:OOAk4}
    \end{subfigure}
    \hfill
    \begin{subfigure}{0.48\linewidth}
    \centering
        \includegraphics[width=\linewidth,trim=20 0 50 27,clip]{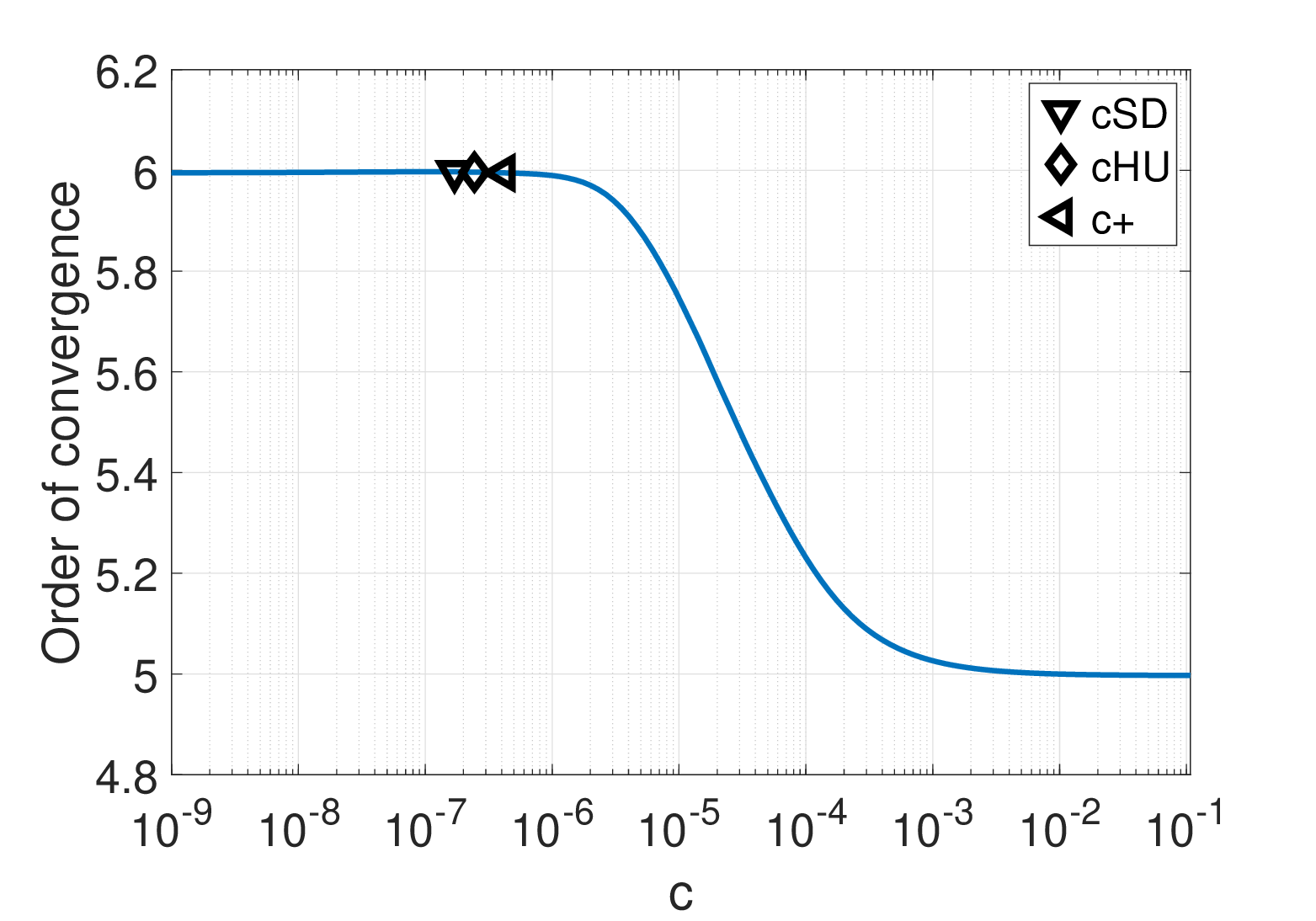}
        \captionsetup{skip=0pt}
        \caption{$k=5, N=32$ to $128$}
        \label{fig:OOAk5}
    \end{subfigure}
    \captionsetup{skip=-1pt}
    \caption{Order of convergence as a function of ESFR parameter $c$}
    \label{fig:OOA_linear_advection}
\end{figure}

\begin{figure}[ht]
    \centering
    \includegraphics[width=0.7\linewidth,trim=20 0 50 27,clip]{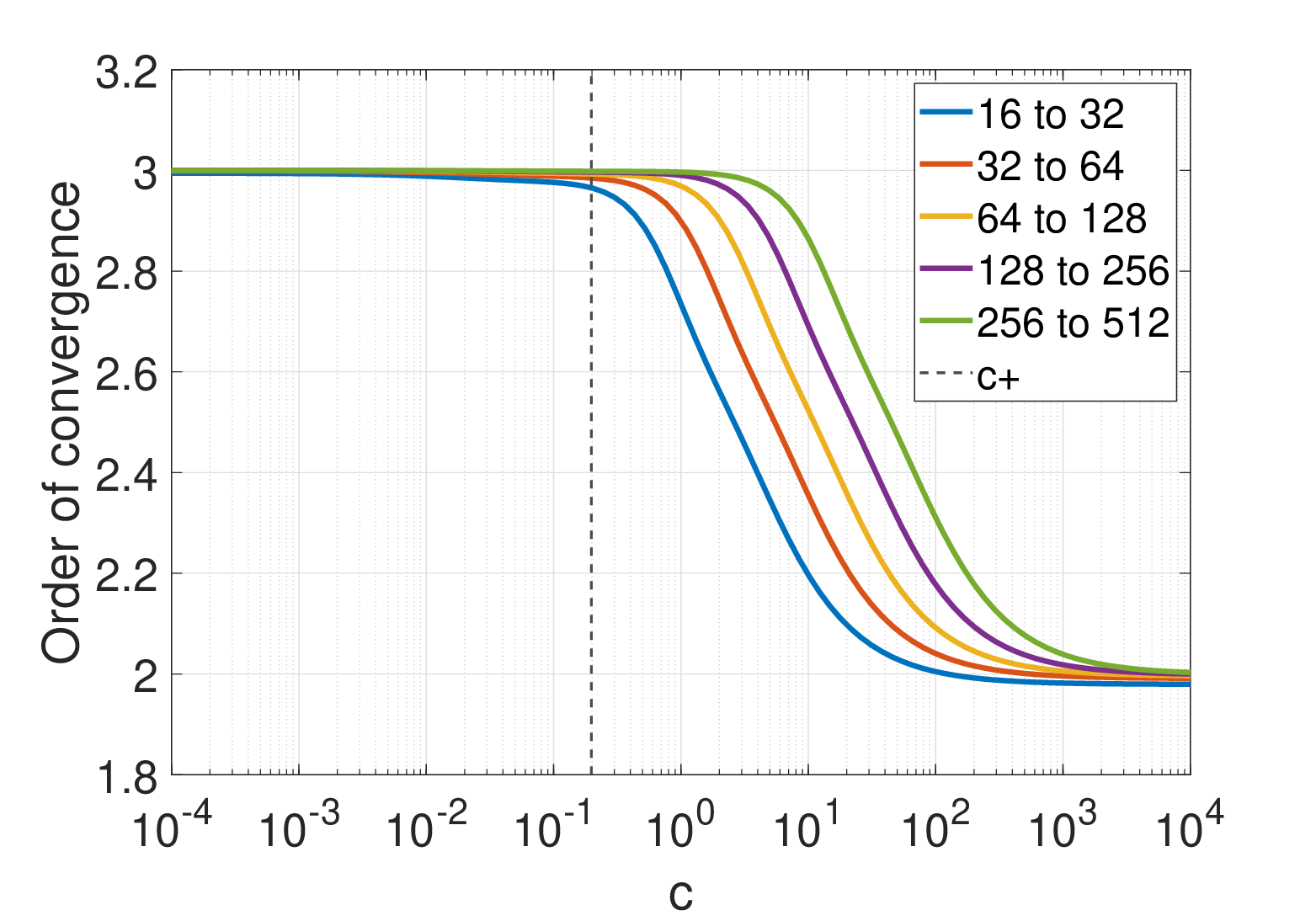}
    \captionsetup{skip=0pt}
    \caption{1D linear advection, $k=2$, convergence for different numbers of elements }
    \label{fig:OOAk2_evolution}
\end{figure}

\subsection{One-dimensional Inviscid Burgers Equation}

In this study, we extend our numerical investigation of the $\mathcal{L}^2$ error and order of convergence to the one-dimensional inviscid Burgers equation. While the error estimate derived before is only valid for the linear advection case, we aim to explore the behavior of the error and order of convergence in the nonlinear inviscid Burgers equation.

\subsubsection{Problem}

We now focus on the one-dimensional inviscid Burgers equation with a non-zero source term, which is a nonlinear hyperbolic PDE:
\begin{align}
&\partial_t u + u\partial_x u = \sin(\pi(x-t))\left[1-\cos(\pi(x-t))\right]\text{ in }\Omega\times[0,T]\label{eq:burgers}\\
&u(x,0)=\cos(\pi x)\hspace{3mm}\forall x\in \Omega
\end{align}
In this study, the domain of interest is $\Omega=[0;2]$. We consider the initial condition $u_0=\cos(\pi x)$, the source term has been designed such that the exact solution is simply: $u(x,t) = \cos(\pi(x-t))$. The simulation is run until a final time $T=2$.

We solve the problem numerically using the fourth-order explicit Runge-Kutta (RK4) time-stepping scheme. The Gauss-Lobatto (GL) quadrature points are used for the spatial discretization, and we take small enough time steps to avoid temporal discretization errors. We use the Lax-Friedrichs flux as the interface numerical flux. For these cases we employ our in-house CFD solver, \texttt{PHiLiP}~\cite{doug_shi_dong_2022_6600853}, to perform the numerical simulations.

\subsubsection{$\mathcal{L}^2$-Error Results}

The error estimate is computed using (\ref{eq:error_quadrature_rule}). 
In order to ensure the accuracy of the numerical integration, the error is overintegrated by adding 10 additional quadrature points to the $k+1$ solution points.

\begin{figure}[htp!]
    \centering
    \begin{adjustbox}{width=\textwidth,center}
    \begin{subfigure}{0.5\linewidth}
    \centering
        \includegraphics[width=\linewidth,trim=13 0 60 0,clip]{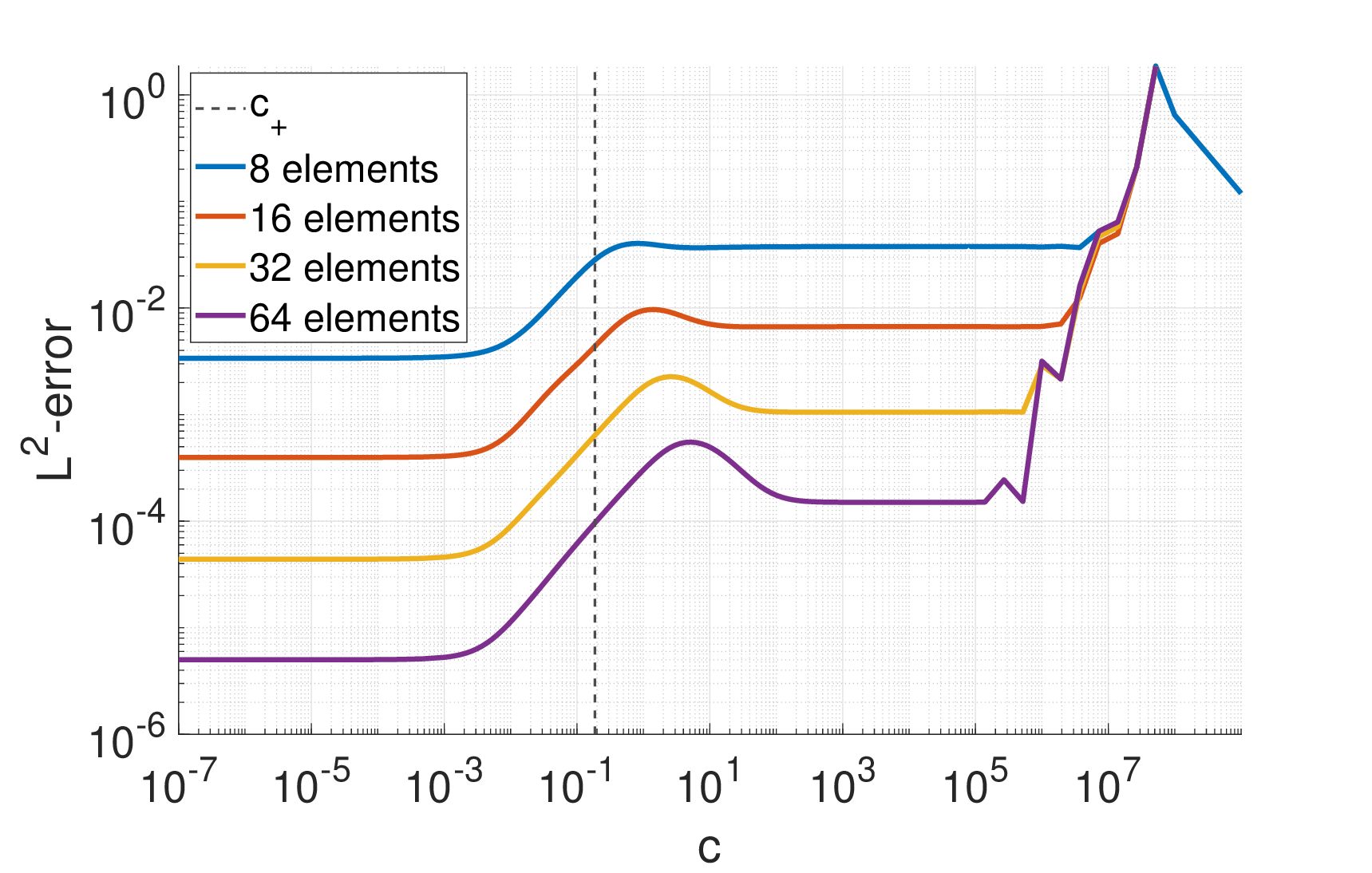}
        \captionsetup{skip=0pt}
        \caption{$k=2, N=8$ to $64$}
        \label{fig:L2_error_k2}
    \end{subfigure}
    \hfill
    \begin{subfigure}{0.5\linewidth}
    \centering
        \includegraphics[width=\linewidth,trim=13 0 60 0,clip]{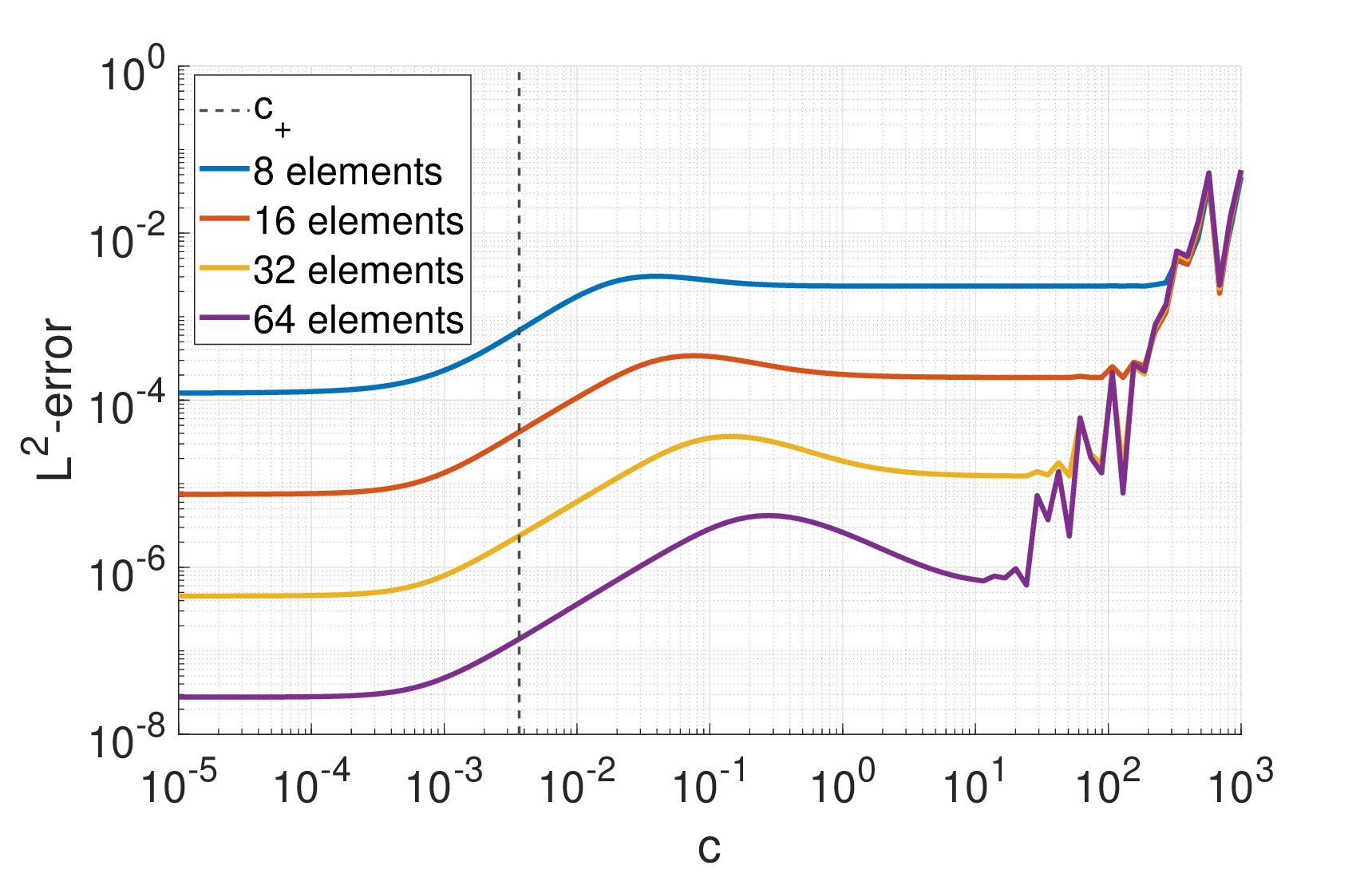}
        \captionsetup{skip=0pt}
        \caption{$k=3, N=8$ to $64$}
        \label{fig:L2_error_k3}
    \end{subfigure}
    \end{adjustbox}
    
    \vspace{-2mm}
    
    \begin{adjustbox}{width=\textwidth,center}
    \begin{subfigure}{0.5\linewidth}
    \centering
        \includegraphics[width=\linewidth,trim=13 0 57 0,clip]{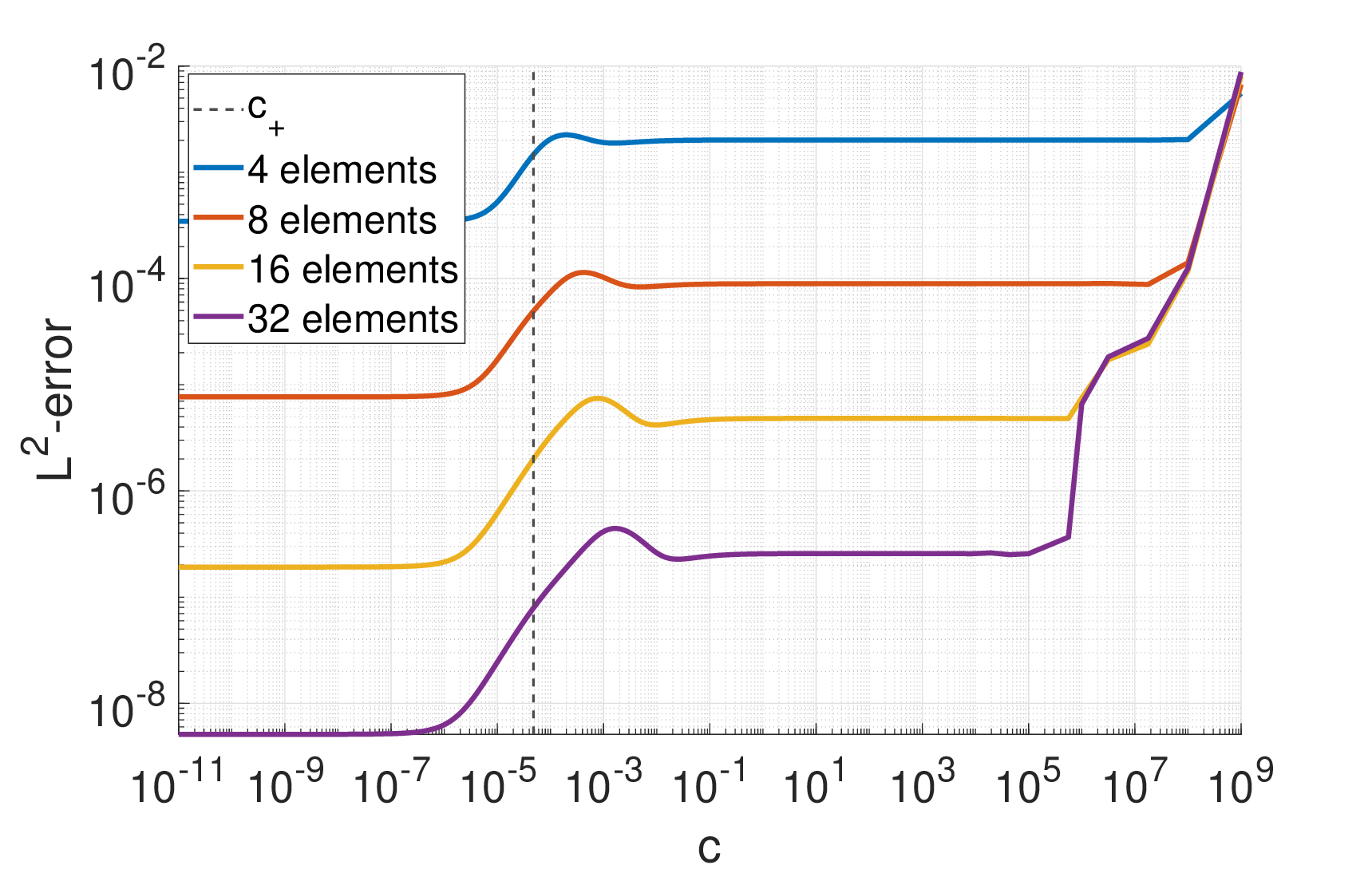}
        \captionsetup{skip=0pt}
        \caption{$k=4, N=4$ to $32$}
        \label{fig:L2_error_k4}
    \end{subfigure}
    \hfill
    \begin{subfigure}{0.5\linewidth}
    \centering
        \includegraphics[width=\linewidth,trim=7 0 70 0,clip]{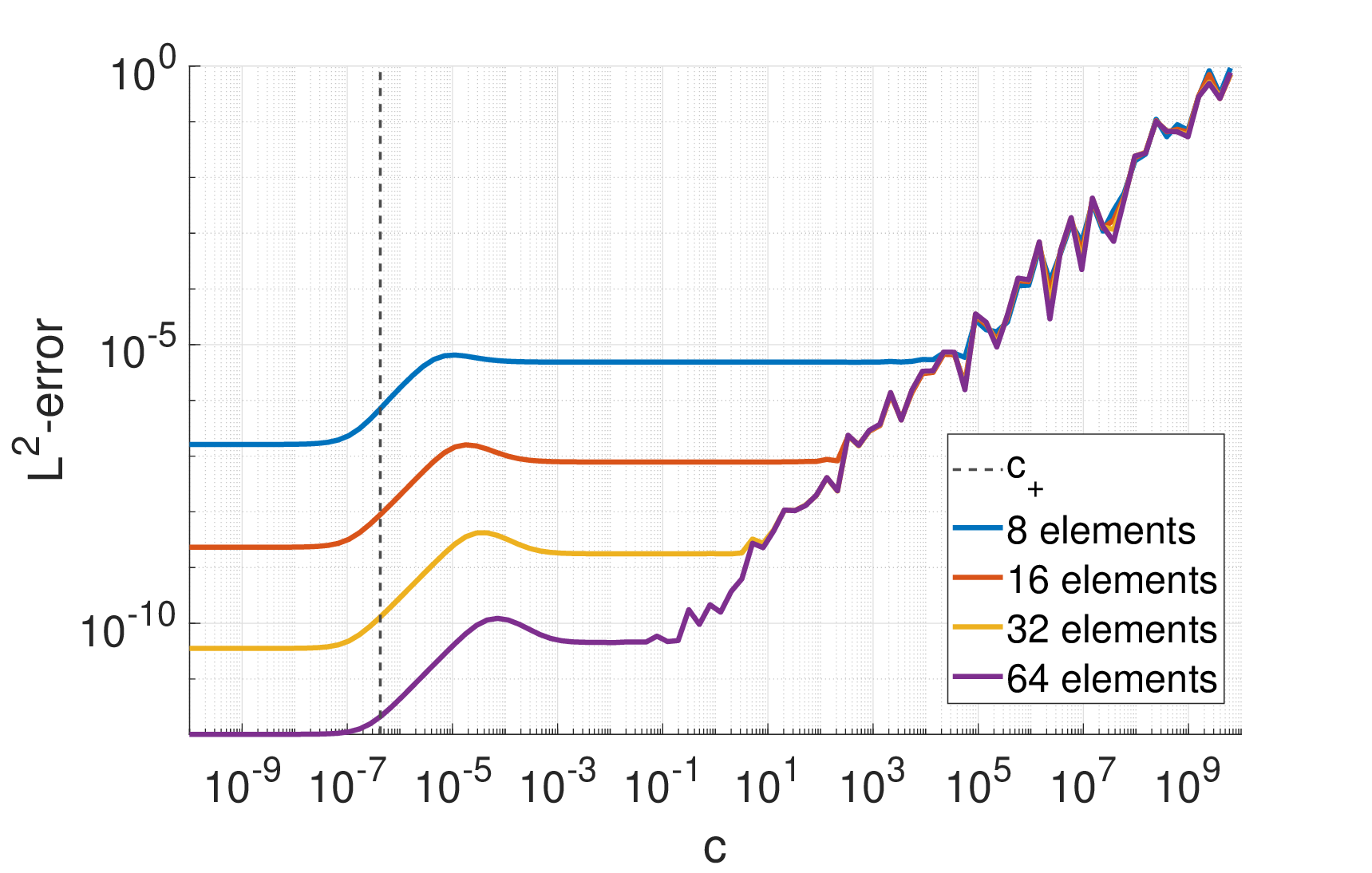}
        \captionsetup{skip=0pt}
        \caption{$k=5, N=8$ to $64$}
        \label{fig:L2_error_k5}
    \end{subfigure}
    \end{adjustbox}
    \captionsetup{skip=0pt}
    \caption{1D Burgers $\mathcal{L}^2$-error as a function of ESFR parameter $c$}
    \label{fig:L2_error_1D_Burgers}
\end{figure}

Figure~\ref{fig:L2_error_1D_Burgers} shows the $\mathcal{L}^2$-error as a function of the $c$ parameter for polynomial orders $k=2$ to 5. The behavior of the $\mathcal{L}^2$-error is similar to that observed for the linear advection (fig.\ref{fig:error_p3_N32}-\ref{fig:error_p3_N256}), since in both cases, the error is roughly constant in $c$ for both low and higher values of $c$ with an increase of degree of magnitude around $c_+$. However, contrary to the linear advection case, this jump is non-monotone for the Burgers' problem; where we observe a local maximum followed by a decrease before reaching a new plateau. At very high values of $c$, the error grows non-monotonically. As the number of degrees of freedom grows, through an increase in the number of elements, the sooner we observe the abrupt rise in the error. At a certain point, all errors take the same value regardless of the number of elements. This general behavior is common for all polynomial degrees, except for $k=3$ and 5, where the error starts to rise at lower values of $c$ compared to both $k=2$ and 4.

\subsubsection{Order of Convergence Results}

The order convergence is computed in the same way as for the linear advection case. Figure~\ref{fig:OOA_1D_Burgers} shows the order of convergence results for polynomial orders 2 to 5. At low values of $c$ the order of convergence is close to $k+1$ then there is a local minimum that reaches the $k$-th order before another plateau at about $k+1/2$. Contrary to the linear advection case, for high values of $c$, the order of convergence is reduced to zero. This break occurs around $c=10^1$ for odd polynomial degrees and $c=10^5$ for even.

\begin{figure}[htp!]
    \centering
    \begin{adjustbox}{width=\textwidth,center}
    \begin{subfigure}{0.5\linewidth}
    \centering
        \includegraphics[width=\linewidth,trim=30 0 65 0,clip]{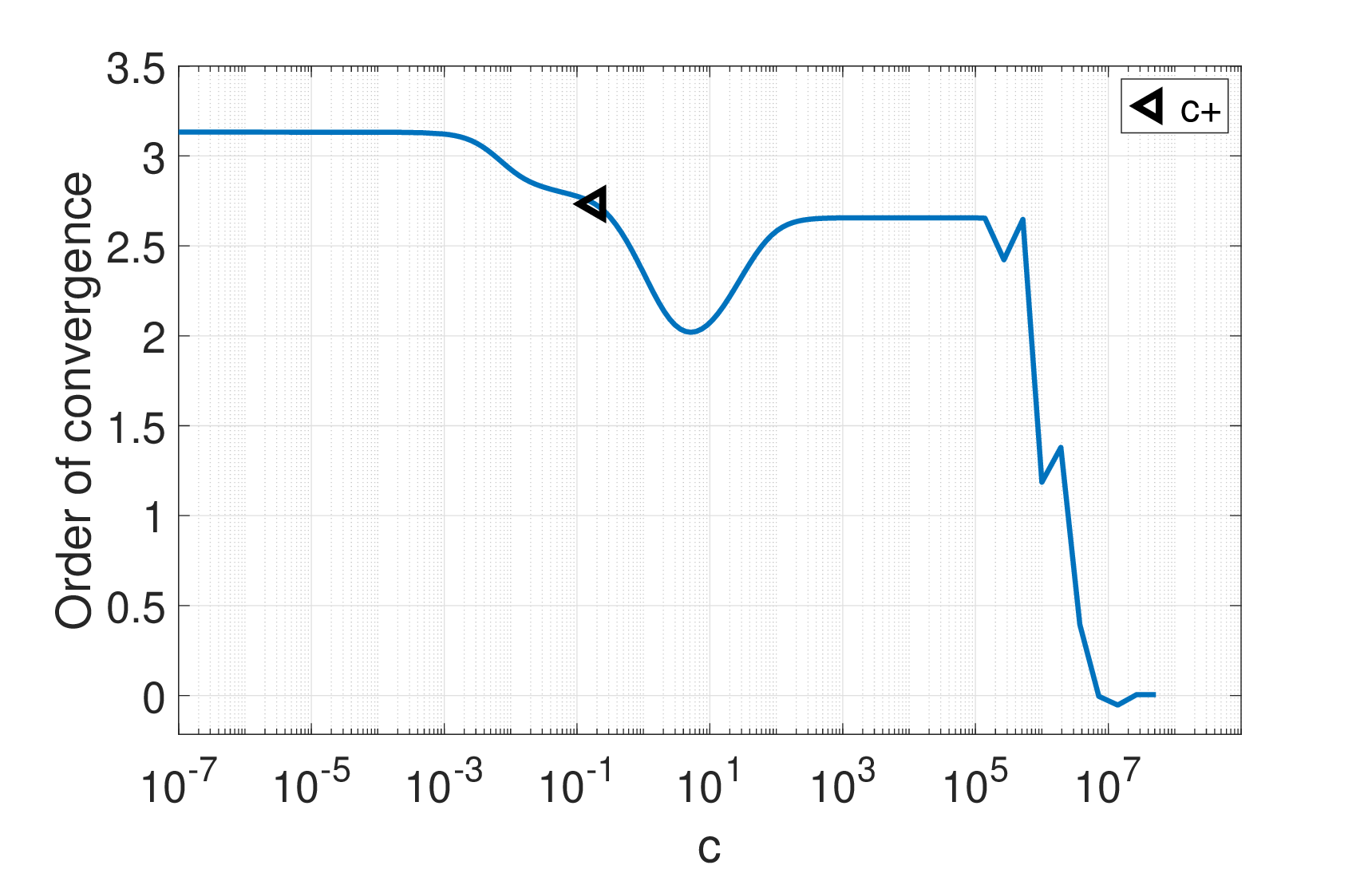}
        \captionsetup{skip=0pt}
        \caption{$k=2, N=8$ to $64$}
        \label{fig:OOAk2_burgers}
    \end{subfigure}
    \hfill
    \begin{subfigure}{0.5\linewidth}
    \centering
        \includegraphics[width=\linewidth,trim=30 0 58 0,clip]{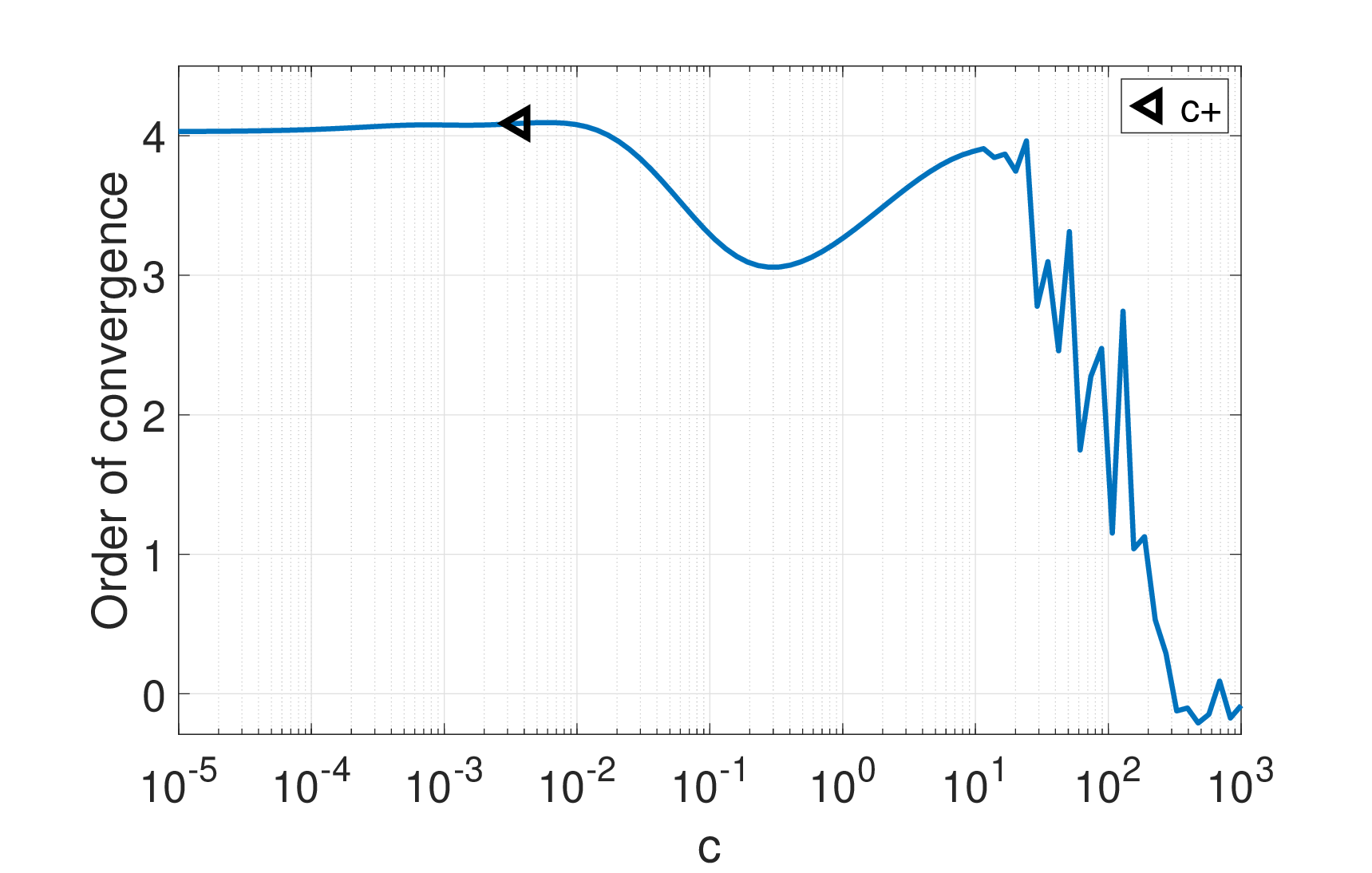}
        \captionsetup{skip=0pt}
        \caption{$k=3, N=8$ to $64$}
        \label{fig:OOAk3_burgers}
    \end{subfigure}
    \end{adjustbox}
    
    \vspace{-2mm}
    
    \begin{adjustbox}{width=\textwidth,center}
    \begin{subfigure}{0.5\linewidth}
    \centering
        \includegraphics[width=\linewidth,trim=30 0 58 0,clip]{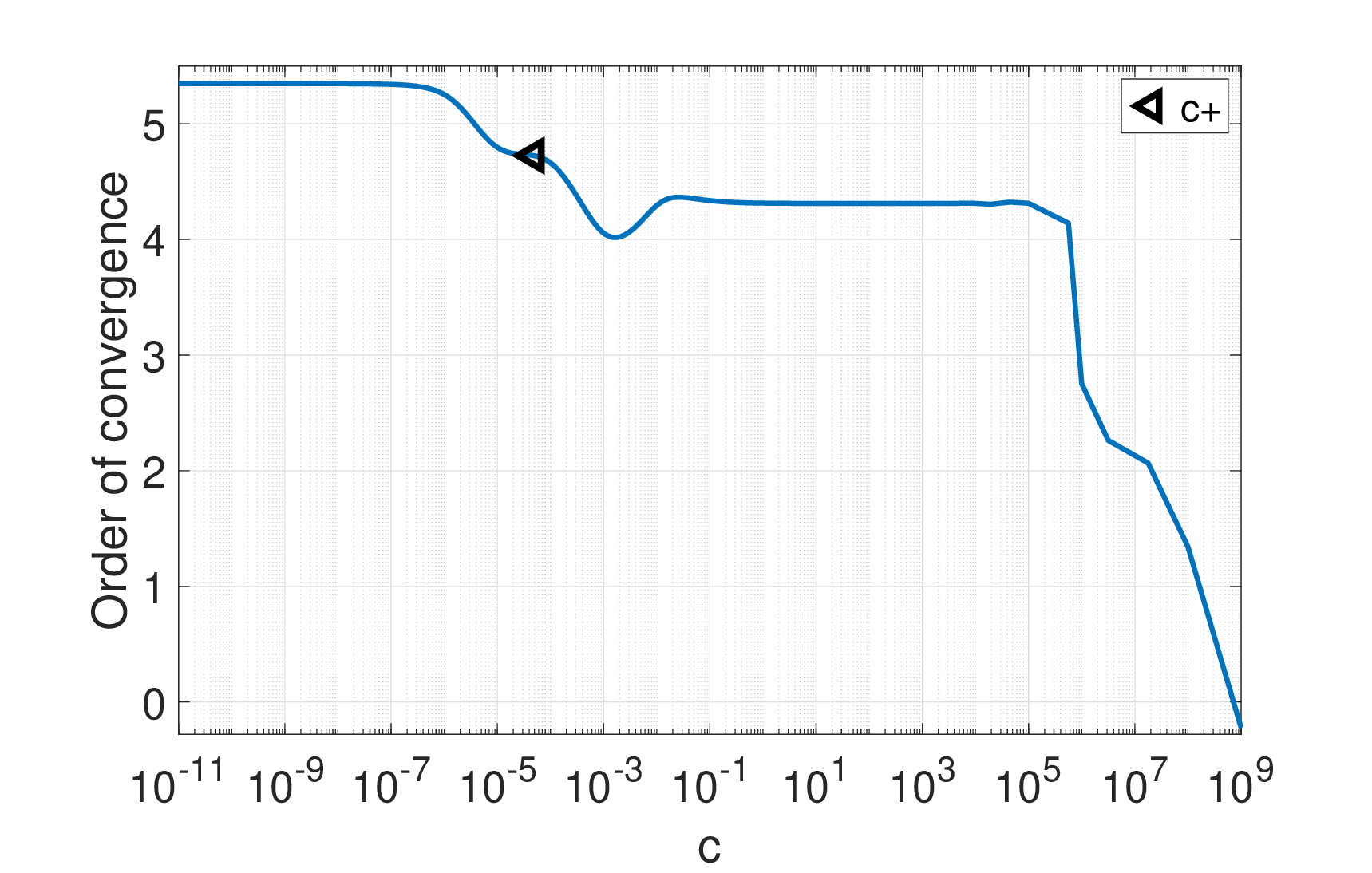}
        \captionsetup{skip=0pt}
        \caption{$k=4, N=4$ to $32$}
        \label{fig:OOAk4_burgers}
    \end{subfigure}
    \hfill
    \begin{subfigure}{0.5\linewidth}
    \centering
        \includegraphics[width=\linewidth,trim=30 0 60 0,clip]{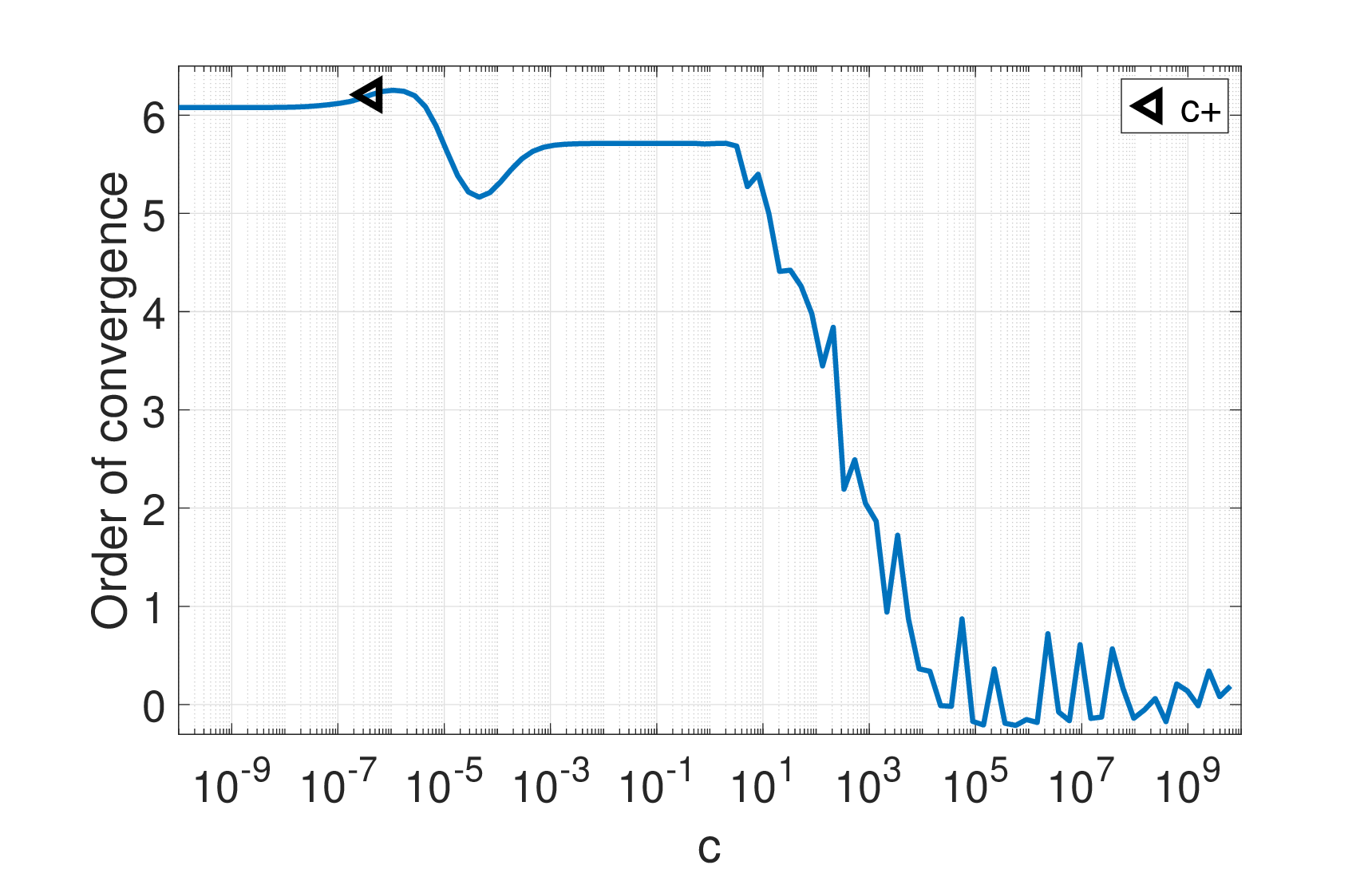}
        \captionsetup{skip=0pt}
        \caption{$k=5, N=8$ to $32$}
        \label{fig:OOAk5_burgers}
    \end{subfigure}
    \end{adjustbox}
    \captionsetup{skip=0pt}
    \caption{1D Burgers, order of convergence as a function of ESFR parameter $c$}
    \label{fig:OOA_1D_Burgers}
\end{figure}

\subsection{Two-dimensional Euler Equation}

While the one-dimensional cases provide valuable insights, extending the analysis to the two-dimensional Euler equation allows for a more comprehensive evaluation of the method's performance in realistic fluid dynamics simulations. By investigating the order of convergence in this context, we can gain a deeper understanding of the method's behavior and build confidence in its accuracy when applied to more practical engineering problems involving two-dimensional flows.

\subsubsection{Governing Equations}

This section describes a numerical experiment on a flow governed by the two-dimensional Euler equations. These equations can be written as a system of equations:

\begin{equation}
    \partial_tU+\partial_xF+\partial_yG=0
\end{equation}
with:
\begin{equation}
    U=\left(\begin{array}{c}
\rho \\
\rho u \\
\rho v \\
\rho \mathrm{e} 
\end{array}\right), \quad F=\left(\begin{array}{c}
\rho u \\
\rho u^2+\mathrm{p} \\
\rho u v \\
\rho u\left(\mathrm{e} +\frac{\mathrm{p}}{\rho}\right)
\end{array}\right), \quad G=\left(\begin{array}{c}
\rho v \\
\rho u v \\
\rho v^2+\mathrm{p} \\
\rho v\left(\mathrm{e} +\frac{\mathrm{p}}{\rho}\right)
\end{array}\right)
\end{equation}
where $u$ and $v$ are velocity components, $\rho$ is the density, $\mathrm{e}$ is the total energy per unit mass and the pressure $\mathrm{p}$ is computed with the equation of state:
\begin{equation}
    \mathrm{p}=(\gamma-1) \rho\left(\mathrm{e} -\frac{1}{2}\left(u^2+v^2\right)\right)
\end{equation}
where the ratio of specific heats $\gamma=1.4$ for air.

\subsubsection{Isentropic Vortex}
The isentropic Euler vortex problem serves as a benchmark to evaluate the order of accuracy of numerical methods due to the availability of exact solutions throughout the simulation. This particular test case holds great significance for high-order methods, as their theoretical accuracy should enable them to sustain the propagation of the vortex while minimizing the production of entropy over an extended duration. The ability to maintain a stable vortex structure without introducing unwanted numerical dissipation is particularly valuable in simulating turbulent flows~\cite{spiegel2015survey}.

The isentropic vortex test case has first been used by Shu~\cite{cockburn1998essentially} to compare the finite volume (FV) scheme with its high-order variation: weighted essentially non-oscillary (WENO). Later, this test was employed by Vincent~\cite{vincent2011insights} and Castonguay~\cite{Castonguay} in the ESFR context. There are several variations in the formulation of the isentropic vortex problem. We will use Spiegel's definition, which aims to unify most of them. He published an article to illustrate some limits of this numerical test and provides recommendations~\cite{spiegel2015survey}.

The initial conditions for the vortex problem are established by adding perturbations in velocity $(u,v)$ and temperature $\theta$ to a uniform mean flow. These perturbations are carefully designed to maintain a constant entropy, $S = p/\rho^\gamma$, across the entire computational domain. These perturbations are derived from the following Gaussian function:
\begin{equation}
    g(x,y)=\beta\exp\left(-\frac{1}{2 \sigma^2}\left[\left(\frac{x}{l}\right)^2+\left(\frac{y}{l}\right)^2\right]\right),
    \label{eq:Gaussian_perturbation}
\end{equation}
where $\beta$ is the maximum strength of the perturbation and the parameter $\sigma$ represents the standard deviation that determines the width of the vortex, while $l$ corresponds to a characteristic length scale.

Using the Gaussian function (\ref{eq:Gaussian_perturbation}), the perturbation in velocity and temperature can be expressed as follows:

\begin{minipage}{0.25\textwidth}
    \begin{equation}
        \Delta u =-\frac{y}{l} g
    \end{equation}
\end{minipage}
\hfill
\begin{minipage}{0.25\textwidth}
    \begin{equation}
        \Delta v =+\frac{x}{l} g
    \end{equation}
\end{minipage}
\hfill
\begin{minipage}{0.35\textwidth}
    \begin{equation}
        \Delta \theta =-\frac{(\gamma-1)}{2 \sigma^2} g^2
    \end{equation}
\end{minipage}

Finally, using these perturbation functions, the initial conditions in terms of primitive variables are given by:

\begin{minipage}{0.40\textwidth}
    \begin{align}
        \tilde{u}_{0} & =M_{\infty} \cos \alpha+\Delta u \label{eq:Euler_IC_u}\\
        \tilde{v}_{0} & =M_{\infty} \sin \alpha+\Delta v \label{eq:Euler_IC_v}
    \end{align}
\end{minipage}
\hfill
\begin{minipage}{0.40\textwidth}
    \begin{align}
        \tilde{\rho}_0 & =(1+\Delta \theta)^{\frac{1}{\gamma-1}} \label{eq:Euler_IC_rho}\\
        \tilde{p}_0 & =\frac{1}{\gamma}(1+\Delta \theta)^{\frac{\gamma}{\gamma-1}}\label{eq:Euler_IC_p}
    \end{align}
\end{minipage}


In (\ref{eq:Euler_IC_u}), (\ref{eq:Euler_IC_v}), the Mach number is denoted by $M$ with the subscript $\infty$ representing the mean flow quantity. The parameter $\alpha$, represents the angle of attack in the mean flow and dictates the direction of propagation of the vortex. The subscript $0$ signifies initial condition quantities. The tilde accents on the primitive variables indicate their nondimensional nature, where $\rho_{\infty}, a_{\infty}, \text { and } \theta_{\infty}$ are chosen as characteristic scales for density, velocity, and temperature, respectively. Table~\ref{tab:isentropic_vortex_parameter} summarizes all parameters used to set up the isentropic vortex simulation.

\begin{table}[t!]
    \centering
    \begin{tabular}{ccccccccccc}
        \toprule
        $\alpha$ & $\gamma$ & $M_{\infty}$ & $\rho_{\infty}$ & $\mathrm{p}_{\infty}$ & $\theta_{\infty}$ & $l$ & $\sigma$ & $\beta$ &$\Omega$ & $T$\\
        \midrule[0.2pt]
        $45^{\circ}$ & 1.4 & $\sqrt{\frac{2}{\gamma}}$ & 1 & 1 & 1 & 1 & 1 & $M_{\infty} \frac{5 \sqrt{2}}{4 \pi} \mathrm{e}^{\frac{1}{2}}$ & $[-10, 10]^2$ & 1.4790\\
        \toprule
    \end{tabular}
    \captionsetup{skip=0pt}
    \caption{Parameters defining the isentropic vortex problem}
    \label{tab:isentropic_vortex_parameter}
\end{table}

Periodic boundary conditions are applied to the domain boundaries to mimic an infinite domain. According to Spiegel~\cite{spiegel2015survey}, this statement is valid only for a sufficiently large computational domain and sufficiently small final time, otherwise periodic boundary conditions all around the domain leads to an infinite array of coupled vortex instead of a single vortex in an infinite space. In addition, the local errors spread over the domain as a wave, reflect against the periodic boundary and accumulate in the domain. To prevent the coupling effect, we use a computational domain $\Omega=[-10, 10]^2$, and we terminate the simulation before the wave errors reach the boundary. The final simulation time is 1.4790 unit time, which represents $1/16$ of a period which is the time for the vortex to advect around the entire domain and return to the origin.

The numerical simulations are advanced in time using a fourth-order explicit Runge-Kutta scheme RK4 with a sufficiently small time step to ensure the temporal discretization does not impact significantly the error. To assess the convergence of the numerical solution, a grid of size $2^n\times 2^n$ is employed, where $n$ varies between 2 and 6 and is adapted to the polynomial order. The Gauss-Lobatto (GL) quadrature points are used for the spatial discretization, and we use the Roe~\cite{roe1981approximate} flux as the interface numerical flux.

The exact solution to this problem is simply the initial vortex convected by the mean flow. The $\mathcal{L}^2$-error is computed using (\ref{eq:error_quadrature_rule}) and the OOA is measured by computing the $\mathcal{L}^2$-error for each grid size and calculating the average slope of the $\mathcal{L}^2$-error versus the width of elements on a log scale.

 \subsubsection{Order of Convergence Results}

 \begin{figure}[htp!]
    \centering
    \begin{subfigure}{0.495\linewidth}
    \centering
        \includegraphics[width=\linewidth,trim=30 0 58 0,clip]{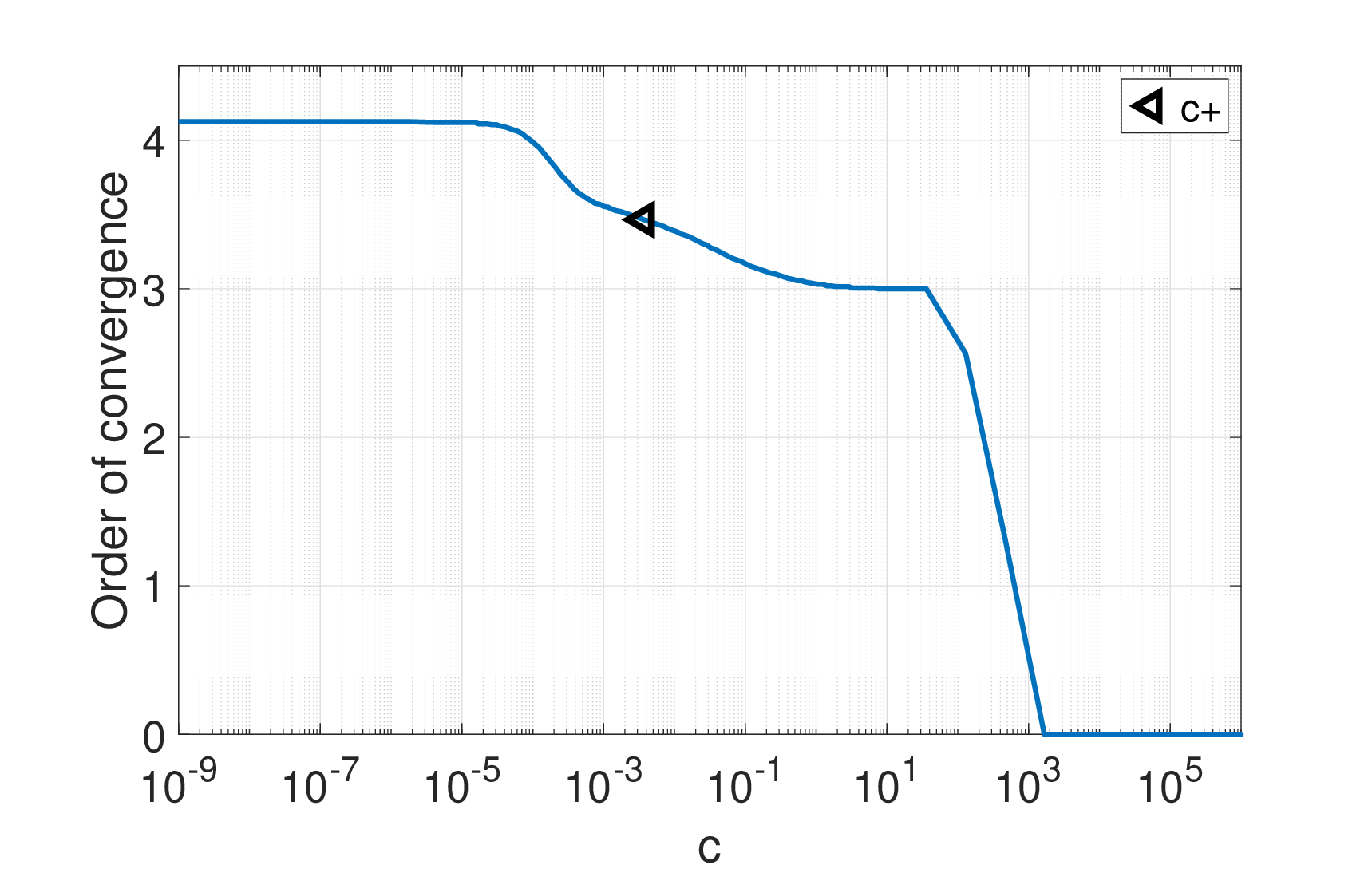}
        \captionsetup{skip=0pt}
        \caption{$k=3$, {\color{black} $N=32^2$, $64^2$ and $128^2$}}
        \label{fig:EulerOOAk3}
    \end{subfigure}
    \hfill
    \begin{subfigure}{0.495\linewidth}
    \centering
        \includegraphics[width=\linewidth,trim=30 0 58 0,clip]{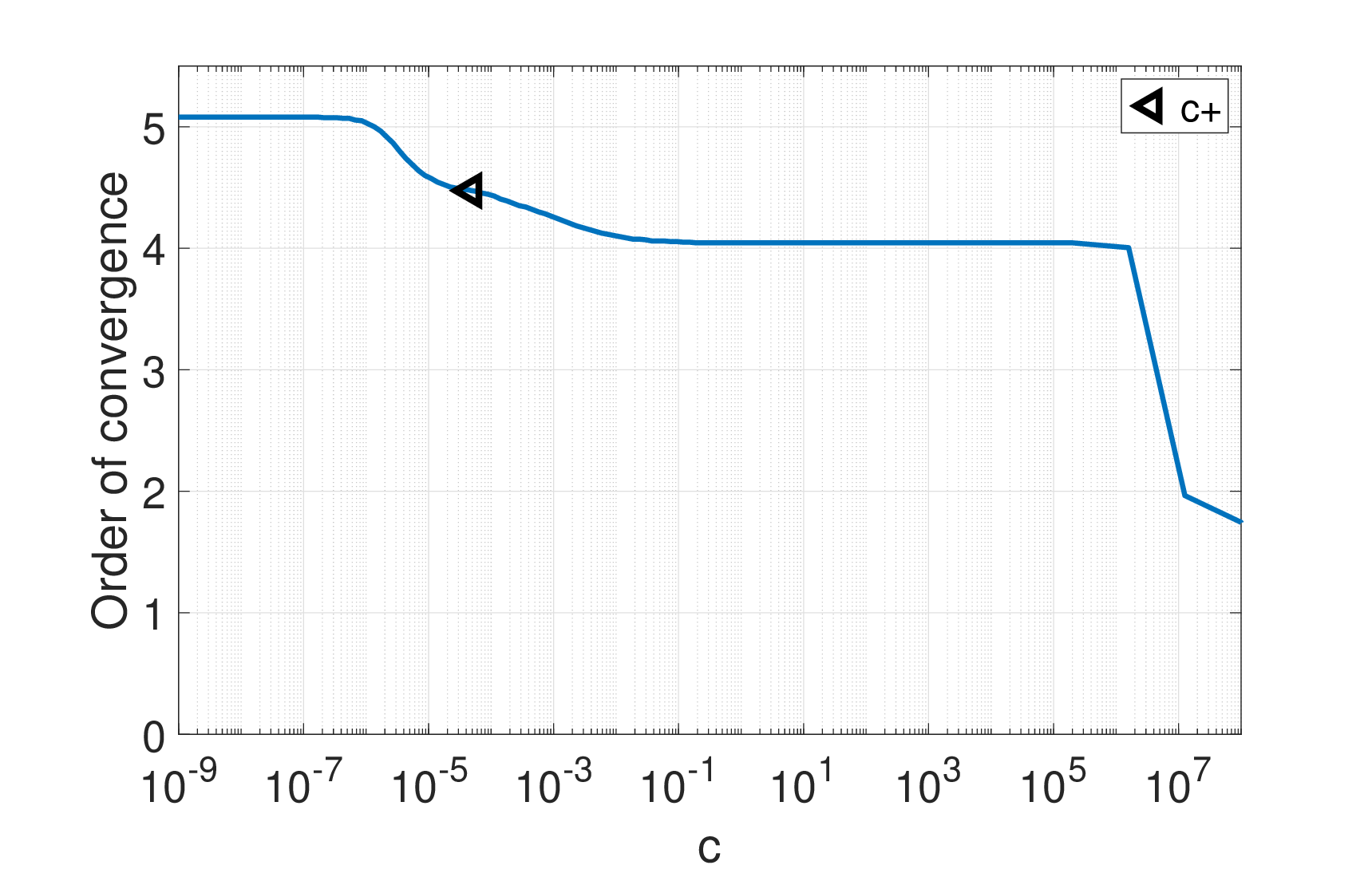}
        \captionsetup{skip=0pt}
        \caption{$k=4$, {\color{black} $N=32^2$, $64^2$ and $128^2$}}
        \label{fig:EulerOOAk4}
    \end{subfigure}
    \captionsetup{skip=0pt}
    \caption{2D Euler, isentropic vortex, order of convergence as a function of ESFR parameter $c$}
    \label{fig:OOA_2D_Euler}
\end{figure}

The order of convergence results for polynomial orders 3 and 4 are depicted in figure~\ref{fig:OOA_2D_Euler}. These results exhibit a similar trend to what we observed in the case of the one-dimensional linear advection. Specifically, we observe a convergence order of $k+1$ for low values of $c$ and this order decreases by one for higher values of $c$. However, there is a difference in that the loss of convergence order occurs earlier, before reaching the $c_+$ value. Additionally, at a certain point, the order of convergence collapses to zero, mirroring the behavior observed in the case of the one-dimensional Burgers' equation. Additionally, it is worth noting that, similar to the case of Burgers' equation, the order of convergence deteriorates to zero earlier for the third polynomial order compared to the fourth orders. More specifically, this occurs at approximately $c=10^1$ for the third order, whereas for the fourth-orders, it appears at approximately $10^5$. The $c_+$ value has been rigorously determined only for the linear advection, but since it is also used in practice for nonlinear problems, it is shown here.

\section{Conclusion}
\label{chap:conclusion}

In this work, an $h$-version of a priori error estimate for the ESFR scheme has been established. Our numerical experiments demonstrated that the error bound we derived matches the correct order of convergence for the linear advection. Moreover, the error estimate showed the role of the $c$ parameter on the order of convergence, but due to unknown constants, it is not possible to determine a value of the ESFR parameter $c$ where the order is lost.

Based on numerical experiments, it was observed that the $k+1$ order of convergence remains valid for certain non-linear problems in both one and two dimensions, particularly at low values of the $c$ parameter. However, as the parameter increases, the convergence order is limited to the order of approximation until it eventually diminishes to zero.

{\color{black}

Finally, this paper has focused on proving the error estimate for the ESFR scheme. However, it may be possible to extend this proof to the more general Sobolev flux reconstruction method (GSFR). This method, originally derived by Trojak~\cite{trojak2018generalised}, is very similar to the ESFR scheme except that it is proved to be stable with respect to the full Sobolev norm so that every derivative up to order $k$ are included.
It is a $k$ parameters scheme; allowing for greater flexibility to filter the discontinuous Galerkin mass matrix, however, it is not clear as of yet if this method allows for a further increase in the time step over and above the ESFR scheme. This extension would provide a more comprehensive understanding of the accuracy of the method and its potential for wider applications.}

\appendix
\renewcommand\thesection{\Alph{section}}
\section*{Appendices}
\section{Proof of lemma \ref{th:equivalent_norms}.}
\label{app:equivalent_morms}
The original proof of theorem \ref{th:equivalent_norms} can be found in \cite{bramble1970estimation}, it is restated here to ease the reading of the thesis. The proof uses two lemmas that are proved in \cite{morrey2009multiple}.

\begin{lemma} \cite[Lem. 1]{bramble1970estimation} or \cite[Th 3.6.10]{morrey2009multiple} For any $f\in \mathcal{H}^{k+1}(I_j)$ there is a unique polynomial p of degree less than or equal to $k$ such that $\int_{I_j} D^\alpha(f+p)dx=0$ for all $\alpha$ with $0\leq |\alpha| \leq k$.
\label{lem:BHlem1}
\end{lemma}

\begin{lemma} \cite[Lem. 2]{bramble1970estimation} or \cite[Th 3.6.11]{morrey2009multiple} Let $I_j$ satisfy a strong cone condition. Then (since $I_j$ is contained in a sphere of radius $\Delta_j$) $|f|_{p,i, I_j} \leq \zeta\Delta_j^{k+1-i}|f|_{p, k+1, I_j}$  for $0 \leq i \leq k$  for all $f \in H^{k+1}(I_j)$ such that the average over $I_j$ of $D^\alpha f$ is 0 for $0\leq|\alpha|\leq k$, where $\zeta$ is a constant independent of $\Delta_j$ and $f$.
\label{lem:BHlem2}
\end{lemma}

\begin{remark}
Morrey supposes a strongly Lipschitz domain, although the proof remains identical if the domain fulfills a strong cone condition.
\end{remark}

We are now in a position to demonstrate Theorem \ref{th:equivalent_norms}, by lemma \ref{lem:BHlem1} we chose $\Bar{p}\in P^k$ such that
\begin{equation*}
    \int_{I_j} D^\gamma(f+\Bar{p})dx=0 \text{ for } |\gamma|\leq k.
\end{equation*}
Therefore, employing lemma \ref{lem:BHlem2}, it can be deduced that
\begin{align*}
    \|f+\bar{p}\|_{p, k+1, I_j} &\leq \zeta\Delta_j^{k+1}|f+\bar{p}|_{p, k+1, I_j}\\
    &=\zeta\Delta_j^{k+1}|f|_{p, k+1, I_j}.
\end{align*}
However, given that $\bar{p}\in P^{k}$, it follows that $\|[f]\|_Q \leq\|f+\bar{p}\|_{p, k+1, I_j}$ by definition~\ref{def:normQ} of the norm on quotient space. Hence 
\begin{equation*}
    \|[f]\|_Q \leq \zeta\Delta_j^{k+1}|f|_{p, k+1, I_j} \text{ for } f\in\mathcal{H}^{k+1}(I_j).
\end{equation*}

The other inequality becomes apparent by noticing
\begin{equation*}
    \Delta_j^{k+1}|f+p|_{p, k+1, I_j}=\Delta_j^{k+1}|f|_{p, k+1, I_j},
\end{equation*}
for any $p\in P^k$ from which we obtain :
\begin{equation*}
    \Delta_j^{k+1}|f|_{p, k+1, I_j}\leq\inf_{p\in P^k} \|f+p\|_{p, k+1, I_j}=\|[f]\|_Q.
\end{equation*}
That completes the proof of theorem \ref{th:equivalent_norms}.

\section{Proof of theorem \ref{th:BH_linear_functional}.}
\label{app:BH_linear_functional}
The proof of theorem \ref{th:BH_linear_functional} can be found in \cite{bramble1970estimation}, it is restated here for completeness.

Because $F$ is linear and satisfies the second condition for the theorem (\ref{eq:Bh_second_condition}),
\begin{equation}
    |F(f)|=|F(f+p)| \text{ for all }p\in P^k.
    \label{eq:pf_bh_1}
\end{equation}
With the first condition  (\ref{eq:Bh_first_condition}) and (\ref{eq:pf_bh_1}) we have
\begin{equation}
    |F(f)|\leq C'\|f+p\|_{p, k+1, I_j}.
    \label{eq:pf_bh_2}
\end{equation}
Taking the infinimum over $P^k$ in (\ref{eq:pf_bh_2}) we have
\begin{equation}
    |F(f)|\leq C'\|[f]\|_Q.
\end{equation}
The result is a direct consequence of theorem \ref{th:equivalent_norms},
\begin{equation}
    |F(f)|\leq C'\zeta \Delta_j^{k+1}|f|_{p, k+1, i_j}.
\end{equation}
We recast the constants into $C_1=C'\zeta$.

\section{Proof of theorem \ref{th:markov}.}
\label{app:markov}
The proof of the second Markov type inequality (\ref{eq:Markov2}) is based on Rivière~\cite{ozisik2010constants}, and it is restated here for completeness.
Consider the reference interval $\Omega_s=[-1;1]$ and associate $\mathcal{L}^2$-orthonormal polynomials, namely the normalized Legendre polynomials $\{\hat{\bar{L}}_i\}_{i=0}^{k}$. The reference interval is mapped to the interval $I_j=[x_{j-1/2};x_{j+1/2}]$, by inverse mapping $\Gamma_j^{-1}$ (\ref{eq:inverse_mapping}).

By chain rule and using the scaling argument, we obtain
\begin{align*}
    \|d_xp\|_{L^2(I_j)}&=\left\|d_rpd_xr\right\|_{L^2(I_j)}\\
    &=\left|d_xr\right|\left\|d_rp\right\|_{L^2(I_j)}\\
    &=\frac{2}{|x_{j+1/2}-x_{j-1/2}|}\left|\frac{x_{j+1/2}-x_{j-1/2}}{2}\right|^{1/2}\left\|d_rp\right\|_{L^2(\Omega_s)}\\
    &\leq \left|\frac{2}{x_{j+1/2}-x_{j-1/2}}\right|^{1/2}\sqrt{C_k}\left\|p\right\|_{L^2(\Omega_s)}\\
    &= \frac{2}{|x_{j+1/2}-x_{j-1/2}|}\sqrt{C_k}\left\|p\right\|_{L^2(I_j)},
\end{align*}
where $C_k$ can be determined for a given polynomial order $k$ by solving the following eigenvalue problem for the maximum eigenvalue, 
\begin{equation*}
    \left(d_r\hat{\bar{L}}_n,d_r\hat{\bar{L}}_m\right)_{\Omega_s}p_m=\lambda\left(\hat{\bar{L}}_n,\hat{\bar{L}}_m\right)_{\Omega_s}p_m.
\end{equation*}
Einstein's summation convention applies to recurring indices. The $\mathcal{L}^2$ inner product on $\Omega_s$ is denoted by $\left(\cdot,\cdot\right)_{\Omega_s}$. Defining $\boldsymbol{S}_{nm}=\left(d_r\hat{\bar{L}}_n,d_r\hat{\bar{L}}_m\right)_{\Omega_s}$, $\boldsymbol{M}_{nm}=\left(\hat{\bar{L}}_n,\hat{\bar{L}}_m\right)_{\Omega_s}$ and using orthonormality of the basis gives us $\boldsymbol{M}=\boldsymbol{I}$ where $\boldsymbol{I}$ is the identity matrix. Then, the above problem reduces to a classical eigenvalue problem
\begin{equation*}
    \boldsymbol{S}_{nm}p_m=\lambda p_m.
\end{equation*}
Let $C_k$ be the maximum eigenvalue $\lambda$, then we can write:
\begin{equation*}
    \left\|d_rp\right\|^2_{L^2(\Omega_s)}\leq C_k \|p\|^2_{L^2(\Omega_s)}
\end{equation*}

For $k=1$, the orthonormal basis is,
\begin{align*}
    &\hat{\bar{L}}_0=\frac{\sqrt{2}}{2},\\
    &\hat{\bar{L}}_1=\frac{\sqrt{6}}{2}r,
\end{align*}
and the matrix $\boldsymbol{S}$ is
\begin{equation*}
    \boldsymbol{S}=\begin{bmatrix}
        0 & 0 \\
        0 & 3 
    \end{bmatrix},
\end{equation*}
then $C_1=3$.

For $k=2$, the basis is
\begin{align*}
     &\hat{\bar{L}}_0, \quad \hat{\bar{L}}_1,\\
    &\hat{\bar{L}}_2=\frac{\sqrt{10}}{4}(3r^2-1),
\end{align*}
and the matrix $\boldsymbol{S}$ is
\begin{equation*}
    \mathbf{S}=\left[\begin{array}{ll|l}
0 & 0 & 0 \\
0 & 3 & 0 \\
\hline 0 & 0 & 15
\end{array}\right]
\end{equation*},
then $C_2=15$.

\section{Proof of theorem \ref{th:Gronwall_lemma}.}
\label{app:Gonwall_lemma}
The theorem \ref{th:Gronwall_lemma} has been proved by Barbu ~\cite[Proposition 1.2]{barbu2016differential}, and it is restated here to ease the read of the thesis.

Let $g_\epsilon$ be the function given by
\begin{equation}
    g_\epsilon:=\frac{1}{2}\left(f_0^2+\epsilon^2\right)+\int_a^t\Psi(s)f(s)ds,\quad t\in[0,T],
\end{equation}
where $\epsilon>0$.
By the (\ref{eq:gronwall_lemma_condition}), we have
\begin{equation}
    f^2(t)\leq 2g_\epsilon(t),\quad t\in[0,T].
    \label{eq:pf_gronwall_1}
\end{equation}
Since $d_tg_\epsilon(t)=\Psi(t)|f(t)|, t\in[0,T],$ we obtain together with (\ref{eq:pf_gronwall_1}):
\begin{align*}
    d_tg_\epsilon(t)&\leq\sqrt{2g_\epsilon(t)}\Psi(t),\quad t\in[0,T],\\
    \frac{d_tg_\epsilon(t)}{\sqrt{2g_\epsilon(t)}}&\leq \Psi(t).
\end{align*}
By integration on the interval $[a,t]$ and renaming the dummy variable $s$, we can deduce that
\begin{align*}
    &\int_a^t\frac{d_tg_\epsilon(s)}{\sqrt{2g_\epsilon(s)}}ds=\sqrt{2g_\epsilon(t)}-\sqrt{2g_\epsilon(a)},\\
    &\sqrt{2g_\epsilon(t)}\leq\sqrt{2g_\epsilon(a)}+\int_a^t \Psi(s)ds,\quad t\in[0,T],
\end{align*}
with (\ref{eq:pf_gronwall_1}), we obtain
\begin{equation*}
    |f(t)| \leq\left|f_0\right|+\epsilon+\int_a^t \Psi(s) d s, \quad t \in[0, T]
\end{equation*}
for every $\epsilon>0$, which implies (\ref{eq:gronwall_lemma_result}) and the theorem \ref{th:Gronwall_lemma} is proved.

\bibliographystyle{siamplain}
\bibliography{references}
\end{document}